\documentclass[a4paper,11pt]{article}
\usepackage[utf8]{inputenc}

\usepackage{graphicx}
\usepackage{geometry}
\usepackage{amsthm}
\usepackage{amssymb}
\usepackage{amsmath}
\usepackage{hyperref}
\usepackage{xcolor}
\usepackage{todonotes}
\usepackage{listings}
\usepackage{microtype}
\usepackage{dsfont}
\usepackage{xspace}
\usepackage{thmtools} 
\usepackage{thm-restate}
\usepackage[font=small, labelfont=bf]{caption}
\usepackage[font=small]{subcaption}

\newtheorem{theorem}{Theorem}

\newtheorem{lemma}[theorem]{Lemma}

\newtheorem{corollary}[theorem]{Corollary}

\newtheorem{conjecture}{Conjecture}
\newtheorem{claim}{Claim}
\newtheorem{observation}[theorem]{Observation}
\newtheorem*{claim*}{Claim}
\newcommand{\proofparagraph}[1]{\smallskip\noindent\textbf{#1}~}

\if10     
\usepackage[mathlines]{lineno}
\newcommand*\patchAmsMathEnvironmentForLineno[1]{%
  \expandafter\let\csname old#1\expandafter\endcsname\csname #1\endcsname
  \expandafter\let\csname oldend#1\expandafter\endcsname\csname end#1\endcsname
  \renewenvironment{#1}%
     {\linenomath\csname old#1\endcsname}%
     {\csname oldend#1\endcsname\endlinenomath}}%
\newcommand*\patchBothAmsMathEnvironmentsForLineno[1]{%
  \patchAmsMathEnvironmentForLineno{#1}%
  \patchAmsMathEnvironmentForLineno{#1*}}%
\AtBeginDocument{%
\patchBothAmsMathEnvironmentsForLineno{equation}%
\patchBothAmsMathEnvironmentsForLineno{align}%
\patchBothAmsMathEnvironmentsForLineno{flalign}%
\patchBothAmsMathEnvironmentsForLineno{alignat}%
\patchBothAmsMathEnvironmentsForLineno{gather}%
\patchBothAmsMathEnvironmentsForLineno{multline}%
}
\linenumbers
\fi

\newcommand{\myorcid}[1]{\href{https://orcid.org/#1}{\includegraphics[height=9pt]{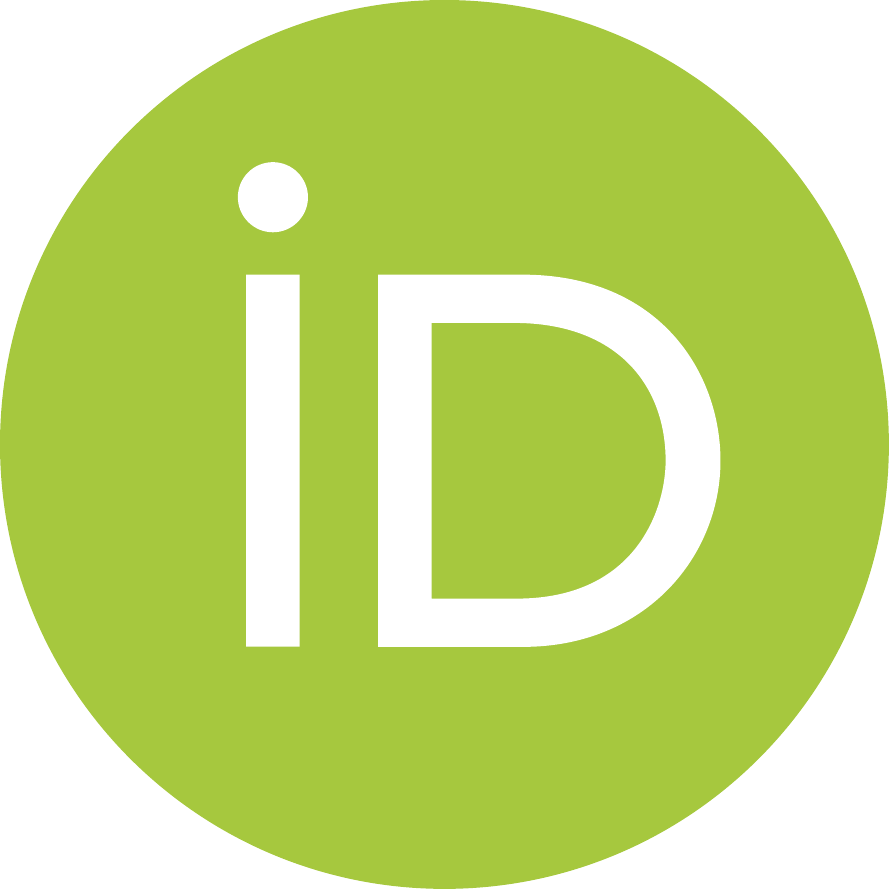}}}

\def\inst#1{$^{#1}$}

\newcommand{\R}{\ensuremath{\mathds{R}}\xspace}
\newcommand{\N}{\ensuremath{\mathds{N}}\xspace}
\renewcommand{\O}{\ensuremath{\mathcal{O}}\xspace}
\DeclareMathOperator{\obs}{obs}
\newcommand{\obsc}{\ensuremath{\obs_\mathrm{c}}}
\newcommand{\fc}{\ensuremath{f_\mathrm{c}}}

\graphicspath{{figures/}}

\theoremstyle{remark}
\newtheorem*{rem}{Remark}
\definecolor{defblue}{rgb}{0.121,0.47,0.705}
\let\emph\relax
\DeclareTextFontCommand{\emph}{\color{defblue}\em}

\date{}

\title{Bounding and Computing Obstacle Numbers
  of~Graphs\thanks{M.~Balko and P.~Valtr were supported by the Grant
    no.~21-32817S of the Czech Science Foundation (GA\v{C}R) and
    support by the Center for Foundations of Modern Computer Science
    (Charles University project UNCE/SCI/004). M.~Balko was supported
    by the European Research Council (ERC) under the European Union's
    Horizon 2020 research and innovation programme (grant agreement
    no.\ 810115). R.~Ganian was supported by the Austrian Science Fund
    (FWF, Project Y1329) and the Vienna Science and Technology Fund
    (WWTF, Project 10.47379/ICT22029).  S.~Gupta was supported by the
    Engineering and Physical Sciences Research Council (EPSRC) grant
    no.\ EP/V007793/1.  M.~Hoffmann was supported by the Swiss
    National Science Foundation within the collaborative DACH project
    \emph{Arrangements and Drawings} as SNSF Project 200021E-171681.
    A.~Wolff was partially supported by the DFG--GA\v{C}R project Wo
    754/11-1.}}

\author{Martin Balko\inst{1}\myorcid{0000-0001-9688-9489}
\and
Steven Chaplick\inst{2}\myorcid{0000-0003-3501-4608}
\and
Robert Ganian\inst{3}\myorcid{0000-0002-3102-4166}
\and
Siddharth Gupta\inst{4}\myorcid{0000-0003-4671-9822}
\and
Michael Hoffmann\inst{5}\myorcid{0000-0001-5307-7106}
\and
Pavel Valtr\inst{1}\myorcid{0000-0002-7762-8045}
\and Alexander Wolff\inst{6}\myorcid{0000-0001-5872-718X}
}

\begin{document}

\maketitle

\begin{center}
{\footnotesize
\inst{1} 
Department of Applied Mathematics, \\
Faculty of Mathematics and Physics, Charles University, Czech Republic \\
\texttt{balko@kam.mff.cuni.cz}\\
\inst{2} 
Maastricht University, The
  Netherlands \\
\inst{3} 
Institute of Logic and Computation, Technische Universit\"at Wien, Austria \\
\texttt{rganian@ac.tuwien.ac.at}\\
\inst{4}
Department of Computer Science \& Information Systems, BITS Pilani Goa Campus, India \\
\texttt{siddharthg@goa.bits-pilani.ac.in
}\\
\inst{5}
Department of Computer Science, ETH Z\"urich, Switzerland \\
\texttt{hoffmann@inf.ethz.ch}\\
\inst{6}
Institut f\"ur Informatik, Universit\"at W\"urzburg, Germany}
\end{center}

\begin{abstract}
An \emph{obstacle representation} of a graph~$G$ consists of a set
  of pairwise disjoint simply-connected closed regions 
  and a one-to-one mapping
  of the vertices of~$G$ to points 
  such that two vertices are adjacent in $G$ if 
  and only if the line segment connecting the two corresponding points does not intersect any obstacle.
  The
  \emph{obstacle number} of a graph is the smallest number of
   obstacles in an obstacle representation of the graph in the plane such that all obstacles are simple polygons.

  It is known that the obstacle number of each $n$-vertex graph is
  $O(n \log n)$ [Balko, Cibulka, and Valtr, 2018] and that there are $n$-vertex graphs whose obstacle number is $\Omega(n/(\log\log n)^2)$ [Dujmovi\'c and Morin, 2015].
  We improve this lower bound to
  $\Omega(n/\log\log n)$ for simple polygons and to $\Omega(n)$ for
  convex polygons. 
  To obtain these stronger bounds, we improve known estimates on the
  number of $n$-vertex graphs with bounded obstacle number, solving a
  conjecture by Dujmovi\'c and Morin.  We also show that if the drawing of some $n$-vertex graph is given as part of the input, then for some drawings $\Omega(n^2)$ obstacles are required to turn them into an obstacle representation of the graph.
  Our bounds are asymptotically tight in several instances.
  
  We complement these combinatorial bounds by two complexity results. First, 
  we show that
  computing the obstacle number of a graph~$G$ is fixed-parameter
  tractable in the vertex cover number of~$G$. 
  Second,
  we show that, given a graph~$G$ and a
  simple polygon~$P$, it is NP-hard to decide whether~$G$ admits an
  obstacle representation using~$P$ as the only obstacle.
\end{abstract}

\section{Introduction}
\label{sec-intro}

An \emph{obstacle} is a simple polygon in the plane.
For a set~$S$ of points in the plane and a set~$\O$ of obstacles, the
\emph{visibility graph} 
of~$S$ with respect to~\O is the
graph with vertex set~$S$ where two
vertices~$s$ and~$t$ are adjacent if and only if the line segment
$\overline{st}$ does not intersect any obstacle in~$\O$. For convenience, we
identify vertices with the points representing
them and edges with the corresponding line segments.  An
\emph{obstacle representation} of a graph~$G$ consists of a set~\O of
pairwise disjoint  
obstacles and a placement
of the vertices of~$G$ such that their visibility graph with respect to~\O
is isomorphic to~$G$. We consider finite point sets and finite
collections of obstacles. 
All obstacles are 
closed. 
For simplicity, we consider point sets to be in \emph{general position}, that is, no three points lie on a common line.

Given a straight-line drawing~$D$ of a graph~$G$, we define the
\emph{obstacle number} of~$D$, $\obs(D)$, to be the smallest number of
obstacles that are needed in order to turn~$D$ into an obstacle
representation of~$G$.  Since every \emph{non-edge} of~$G$ (that is, every
pair of non-adjacent vertices) must be blocked by an obstacle and no
edge of~$G$ may intersect an obstacle, $\obs(D)$ is the cardinality of
the smallest set of faces of~$D$ whose union is intersected by every
non-edge of~$G$.
For a graph~$G$, the \emph{obstacle number} of~$G$, $\obs(G)$, is the
smallest value of~$\obs(D)$, where $D$ ranges over all straight-line drawings of~$G$.
For a positive integer~$n$, let $\obs(n)$ be the maximum
value of $\obs(G)$, where $G$ is any $n$-vertex graph.
The \emph{convex obstacle number} $\obsc(\cdot)$ is defined analogously, except that here all obstacles are convex polygons.

For positive integers $h$ and $n$, let $f(h,n)$ be the number of
graphs on $n$ vertices that have obstacle number at most~$h$.
Similarly, we denote the number of graphs on $n$ vertices that have
convex obstacle number at most~$h$ by $\fc(h,n)$.

Alpert, Koch, and Laison \cite{akl-ong-DCG10} introduced the obstacle
number and the convex obstacle number.  Using Ramsey theory, they
proved that for every positive integer~$k$, there is a (huge) complete
$k$-partite graph with convex obstacle number~$k$.  They also showed
that $\obs(n) \in \Omega(\sqrt{\log n / \log \log n})$.
This lower bound was subsequently
improved to $\Omega(n/\log n)$ by Mukkamala, Pach, and
P\'{a}lv\"{o}lgyi~\cite{mpp-lbong-12} and to
$\Omega(n/(\log \log n)^2)$ by Dujmovi\'c and Morin~\cite{dm-on-15},
who conjectured the following.

\begin{conjecture}[\cite{dm-on-15}]
\label{conj-DujMor}
For all positive integers $n$ and $h$, we have $f(h,n) \in 2^{g(n) \cdot o(h)}$, where
  $g(n) \in O(n \log^2 n)$.
\end{conjecture}

On the other hand, we trivially have $\obs(n) \leq \obsc(n) \leq \binom{n}{2}$ as one can block each non-edge of an $n$-vertex graph with a single obstacle.
Balko, Cibulka, and Valtr~\cite{bcv-dgusno-DCG18} improved this upper bound to $\obs(n) \leq \obsc(n) \leq n \lceil \log{n} \rceil-n+1$, 
refuting a conjecture by Mukkamala, Pach, and P{\'a}lv{\"o}lgyi~\cite{mpp-lbong-12} stating that $\obs(n)$ is around~$n^2$.
For graphs $G$ with chromatic number $\chi$, Balko, Cibulka, and Valtr~\cite{bcv-dgusno-DCG18} showed that $\obs(G) \leq \obsc(G) \leq (n-1)(\lceil\log \chi \rceil + 1)$, which is in $O(n)$ if the chromatic number is bounded by a constant.

Alpert, Koch, and Laison~\cite{akl-ong-DCG10} differentiated between an
\emph{outside obstacle}, which lies in (or simply is) the outer face
of the visibility drawing, and \emph{inside obstacles}, which lie in
(or are) inner faces of the drawing.  They proved that every outerplanar
graph has an outside-obstacle representation, that is, a representation with a single outside obstacle.
Later, Chaplick, Lipp, Park, and Wolff~\cite{clpw-ov1o-GD16} showed that the class of graphs that admit a
representation with a single inside obstacle is incomparable with the class of
graphs that have an outside-obstacle representation.  They found the
smallest graph with obstacle number~2; it has eight vertices and is
co-bipartite.  They also showed that the following \emph{sandwich version} of the
outside-obstacle representation problem is NP-hard: 
Given two graphs~$G$ and~$H$ with
$V(G)=V(H)$ and $E(G)\subseteq E(H)$, is there a graph~$K$ with
$V(K)=V(G)$ and $E(G) \subseteq E(K) \subseteq E(H)$ that admits an
outside-obstacle representation?  Analogous hardness results hold with
respect to inside 
and general obstacles.
Firman, Kindermann, Klawitter, Klemz, Klesen, and
Wolff~\cite{fkkkkw-ooravof-GD22} showed that
every partial 2-tree has an outside-obstacle
representation, which generalizes the
result of Alpert, Koch, and Laison~\cite{akl-ong-DCG10} concerning
the representation of outerplanar graphs.  For (partial)
outerpaths, cactus graphs, and grids, Firman et al.\ constructed
outside-obstacle representations where the vertices are those
of a regular polygon~\cite{fkkkkw-ooravof-GD22}.  Furthermore,
they characterized when the complement of a tree and when a complete
graph minus a simple cycle admits a \emph{convex} outside-obstacle
representation (where the graph vertices are in convex position).

For planar graphs, Johnson and Sar\i\"oz \cite{js-rpslg-CCCG14}
investigated a variant of the problem where the visibility graph~$G$
is required to be plane and a plane drawing~$D$ of~$G$ is given.  
They
showed that computing $\obs(D)$ is NP-hard (by reduction {\em from}
planar vertex cover) and that there is a solution-value-preserving
reduction {\em to} maximum-degree-3 planar vertex cover.  For
computing $\obs(D)$, this reduction yields a polynomial-time
approximation scheme and a fixed-parameter algorithm with respect to $\obs(D)$.
Gimbel, de Mendez, and Valtr~\cite{gpv-onpg-GD17} showed that, for some planar
graphs, there is a large discrepancy between this planar setting and the usual
obstacle number.

A related problem deals with \emph{point visibility graphs}, where the
points are not only the vertices of the graph but also the obstacles
(which are closed in this case). Recognizing point visibility graphs
is contained in $\exists\R$~\cite{gr-sropvg-TCS15} and
is $\exists\R$-hard~\cite{ch-rcpvg-DCG17}, and thus
$\exists\R$-complete.

\paragraph*{Our Contribution.}

Our results span three areas.  First, we improve existing bounds on
the (convex and general) obstacle number (see Section~\ref{sec-lowerBound}), on
the functions $f(n,h)$ and $\fc(n,h)$ (see Section~\ref{sec-complexity}), and
on the obstacle number of drawings (see Section~\ref{sec-drawing}).  Second,
we provide an algorithmic result: computing the obstacle number of a
given graph is fixed-parameter tractable with respect to the vertex
cover number of the graph (see Section~\ref{sec-fpt}).  Third, we investigate
algorithmic lower bounds and show that, given a graph~$G$ and a simple
polygon~$P$, it is NP-hard to decide whether~$G$ admits
an obstacle representation using~$P$ as the only obstacle (see
Section~\ref{sec-hardness}). We now describe our results in more detail. 

First, we prove the currently strongest lower bound on the obstacle
number of $n$-vertex graphs, improving the estimate of
$\obs(n) \in \Omega(n/(\log \log n)^2)$ by Dujmovi\'c and Morin
\cite{dm-on-15}.

\begin{theorem}
  \label{thm-lowerBound}
  There is a constant $\beta>0$ such that, for every $n\in\N$, there exists a
  graph on $n$ vertices with obstacle number at least
      $\beta n/\log \log n$, that is, $\obs(n)\in\Omega(n/\log \log n)$.
\end{theorem}

This lower bound is quite close to the currently best upper bound
$\obs(n) \in O(n\log{n})$ by Balko, Cibulka, and Valtr~\cite{bcv-dgusno-DCG18}.
In fact, Alpert, Koch, and Laison~\cite{akl-ong-DCG10} asked whether the obstacle number of any graph with $n$ vertices is bounded from above by a linear function of $n$.
This is supported by a result of Balko, Cibulka, and Valtr~\cite{bcv-dgusno-DCG18} who proved $\obs(G) \leq \obsc(G) \in O(n)$ for every graph $G$ with $n$ vertices and with bounded chromatic number.
We remark that we are not aware of any argument that would give a strengthening of Theorem~\ref{thm-lowerBound} to graphs with bounded chromatic number.

Next, we show that a linear lower bound holds for convex obstacles.

\begin{theorem}
  \label{thm-lowerBoundConvex}
  There is a constant $\gamma>0$ such that, for every $n\in\N$, there exists a (bipartite)
  graph on $n$ vertices with convex obstacle number at least
  $\gamma n$, that is, 
  $\obsc(n) \in \Omega(n)$.
\end{theorem}

The previously best known bound on the convex obstacle number was $\obsc(G) \in \Omega(n/(\log \log n)^2)$~\cite{dm-on-15}.
We note that the upper bound proved by Balko, Cibulka, and Valtr~\cite{bcv-dgusno-DCG18} actually holds for the convex obstacle number as well and gives $\obsc(n) \in O(n\log{n})$.
Furthermore, the linear lower bound on $\obsc(n)$ from Theorem~\ref{thm-lowerBoundConvex} works for $n$-vertex graphs with bounded chromatic number, asymptotically matching the linear upper bound on the obstacle number of such graphs proved by Balko, Cibulka, and Valtr~\cite{bcv-dgusno-DCG18}; see the remark in Section~\ref{sec-lowerBound}.

The proofs of Theorems~\ref{thm-lowerBound} and~\ref{thm-lowerBoundConvex} are both based on counting arguments that use the upper bounds on $f(h,n)$ and $\fc(h,n)$.
To obtain the stronger estimates on $\obs(n)$ and $\obsc(n)$, we improve the currently best bounds $f(h,n) \in 2^{O(hn\log^2{n})}$ and $\fc(h,n) \in 2^{O(hn\log{n})}$ 
by Mukkamala, Pach, and P\'{a}lv\"{o}lgyi~\cite{mpp-lbong-12} as follows.

\begin{theorem}
  \label{thm-complexity}
  For all positive integers $h$ and $n$, we have
  $f(h,n) \in 2^{O(h n \log n)}$.
\end{theorem}

We also prove the following asymptotically tight upper bound on $\fc(h,n)$.

\begin{theorem}
  \label{thm-complexityConvex}
  For all positive integers $h$ and $n$, we have
  $\fc(h,n) \in 2^{O(n(h+\log n))}$.
\end{theorem}

The upper bound 
in Theorem~\ref{thm-complexityConvex} is asymptotically tight for $h < n$ as Balko, Cibulka, and Valtr~\cite{bcv-dgusno-DCG18} proved that for every pair of integers $n$ and $h$ satisfying $0 < h < n$ we have $\fc(h,n) \in 2^{\Omega(hn)}$. 
Their result is stated for $f(h,n)$, but it is proved using  convex obstacles only.
This matches our bound for $h \in \Omega(\log{n})$.
Moreover, they also 
showed that $\fc(1,n) \in 2^{\Omega(n\log{n})}$, which matches the bound from Theorem~\ref{thm-complexityConvex} also for $h \in O(\log{n})$.

Since 
trivially 
$h \leq \binom{n}{2} \leq n^2$, we get $\log{h} \leq 2\log{n}$, and thus the bound from
Theorem~\ref{thm-complexity} can be rewritten as
$f(h,n) \leq 2^{g(n) \cdot (h/\log{n})} \leq 2^{g(n) \cdot (2h/\log{h})} \in 2^{g(n) \cdot o(h)}$,
where $g(n) \in O(n\log^2{n})$.
Therefore, we get the following corollary, confirming Conjecture~\ref{conj-DujMor}.

\begin{corollary}
\label{cor-complexity}
For all positive integers $h$ and $n$, we have $f(h,n) \in 2^{g(n) \cdot (2h/\log{h})}$,
where $g(n) \in O(n\log^2{n})$.
\end{corollary}

It is natural to ask about estimates on obstacle numbers of fixed drawings, that is, considering the problem of estimating $\obs(D)$.
The parameter $\obs(D)$ has been considered in the literature; e.g., by Johnson and Sar\i\"oz~\cite{js-rpslg-CCCG14}.
Here, we provide a quadratic lower bound on $\max_D \obs(D)$ where the maximum ranges over all drawings of graphs on $n$ vertices.

\begin{theorem}
  \label{thm-drawing}
  There is a constant $\delta>0$  such that, for every $n$, there exists a
  graph $G$ on $n$ vertices and a drawing $D$ of $G$ such that
  $\obs(D) \geq \delta \cdot n^2$.
\end{theorem}

The bound from Theorem~\ref{thm-drawing} is asymptotically tight as we trivially have $\obs(D) \leq \obsc(D) \leq \binom{n}{2}$ for every drawing $D$ of an $n$-vertex graph.
This also asymptotically settles the problem of estimating the convex obstacle number of drawings.

Next, we turn our attention to algorithms for computing the obstacle number. We establish fixed-parameter tractability for the problem when parameterized by the size of a minimum vertex cover of the input graph $G$, called the \emph{vertex cover number} of $G$.

\begin{restatable}{theorem}{thmfpt}
  \label{thm:fpt}
  Given a graph $G$ and an integer $h$, the problem of deciding
  whether $G$ admits an obstacle representation with $h$ obstacles is
  fixed-parameter tractable parameterized by the vertex cover number
  of~$G$.
\end{restatable}

The proof of Theorem~\ref{thm:fpt} is surprisingly non-trivial. On a high level, it uses Ramsey techniques to identify, in a sufficiently large graph $G$, a set of vertices outside a minimum vertex cover which not only have the same neighborhood, but also have certain geometric properties in a hypothetical solution $S$. We then use a combination of topological and combinatorial arguments to show that $S$ can be adapted to an equivalent solution for a graph $G'$ whose size is bounded by the vertex cover number of $G$---i.e., a kernel~\cite{DowneyF13}.

While the complexity of deciding whether a given graph has obstacle
number~1 is still open, the sandwich version of the problem
\cite{clpw-ov1o-GD16} and the version for planar visibility drawings
\cite{js-rpslg-CCCG14} have been shown NP-hard. 
We conclude with a simple new NP-hardness result.

\begin{restatable}{theorem}{hardness}
  \label{thm-hardness}
  Given a graph $G$ and a simple polygon $P$, it is NP-hard to decide
  whether $G$ admits an obstacle representation using $P$ as
  (outside-) obstacle.
\end{restatable}

In the remainder of this paper, for any positive integer~$k$, we will
use $[k]$ as shorthand for $\{1,2,\dots,k\}$.

\section{Improved Lower Bounds on Obstacle Numbers}
\label{sec-lowerBound}

We start with 
the estimate $\obs(n) \in \Omega(n/\log\log n)$
from Theorem~\ref{thm-lowerBound}.
The proof is based on Theorem~\ref{thm-complexity} (cf. Section~\ref{sec-complexity}) 
and the following result by Dujmovi\'{c} and Morin~\cite{dm-on-15}, which shows that an 
upper bound for $f(h,n)$ translates into a 
lower bound for $\obs(n)$.

\begin{theorem}[\cite{dm-on-15}]
\label{thm-blackBox}
For every positive integer $k$, let $H(k) = \max\{h \colon f(h,k) \leq 2^{k^2/4}\}$.
Then, for any constant $c >0$, there exist $n$-vertex graphs with obstacle number at least 
$\Omega\left(\frac{nH(c\log{n})}{c\log{n}}\right)$.
\end{theorem}

\begin{proof}[Proof of Theorem~\ref{thm-lowerBound}]
By Theorem~\ref{thm-complexity}, we have $f(h,k) \in 2^{O(h k \log k)}$ for every positive integer $k$ and thus $H(k) \in \Omega(k/\log{k})$.
Plugging this estimate with $k\in \Omega(\log n)$ into Theorem~\ref{thm-blackBox}, we get the lower bound $\obs(n) \in \Omega(n/\log\log{n})$.
\end{proof}

It remains to prove the lower bound $\obsc(n) \in \Omega(n)$ from Theorem~\ref{thm-lowerBoundConvex}.

\begin{proof}[Proof of Theorem~\ref{thm-lowerBoundConvex}]
By Theorem~\ref{thm-complexityConvex}, we have 
$f_c(h,n)\le 2^{\beta n(h+\log n)}$ for some constant $\beta>0$.
Since the number of graphs with $n$ vertices is $2^{\binom{n}{2}}$, we see that if  $2^{\binom{n}{2}} > 2^{\beta n(h+\log n)}$, then there is an $n$-vertex graph with convex obstacle number larger than $h$.
The inequality $\binom{n}{2} > \beta n(h+\log n)$ is satisfied for $h < \binom{n}{2}/(\beta n) - \log{n} \in \Theta(n)$.
Thus, $\obsc(n) \in \Omega(n)$.
\end{proof}

\begin{rem}
Observe that the proof of Theorem~\ref{thm-lowerBoundConvex} also works if we restrict ourselves to bipartite graphs, as the number of bipartite graphs on $n$ vertices is at least $2^{n^2/4}$.
Thus, the convex obstacle number of an $n$-vertex graph with bounded chromatic number can be linear in~$n$, which asymptotically matches the linear upper bound proved by Balko, Cibulka, and Valtr~\cite{bcv-dgusno-DCG18}.
\end{rem}

\section{On the Number of Graphs with Bounded Obstacle Number}
\label{sec-complexity}

In this section, we prove Theorems~\ref{thm-complexity}
and~\ref{thm-complexityConvex} about improved estimates on the number
of $n$-vertex graphs with obstacle number and convex obstacle number,
respectively, at most~$h$.

\subsection{Proof of Theorem~\ref{thm-complexity}}

We prove that the number $f(h,n)$ of $n$-vertex graphs with obstacle number at most~$h$ is bounded from above by $2^{O(h n \log n)}$. This improves a previous bound of $2^{O(h n \log^2 n)}$ by Mukkamala, Pach, and P{\'a}lv{\"o}lgyi~\cite{mpp-lbong-12} and resolves Conjecture~\ref{conj-DujMor} by Dujmovi\'{c} and Morin~\cite{dm-on-15}.
We follow the approach of Mukkamala, Pach, and P{\'a}lv{\"o}lgyi~\cite{mpp-lbong-12} by encoding the obstacle representations using point set order types for the vertices of the graph and the obstacles.
They used a bound of $O(n\log n)$ for the number of vertices of a single obstacle.
We show that under the assumption that the obstacle representation uses a minimum number of obstacle vertices, this bound can be reduced to $O(n)$.
As a result, the upper bound on $f(h,n)$ improves by a factor of $\Theta(\log n)$ in the exponent.

We say that an obstacle representation $(D,\mathcal{O})$ of a graph $G$ is \emph{minimal} if it contains the minimum number of obstacles and there is no obstacle representation $(D',\mathcal{O}')$ of $G$ with the minimum number of obstacles such that the total number of vertices of obstacles from~$\mathcal{O}'$ is strictly smaller than the number of vertices of obstacles from $\mathcal{O}$. 

Consider a graph~$G$ that admits an obstacle representation using at most~$h$ obstacles. Our goal in the remainder of this section is to show that a minimal obstacle representation of~$G$ can be encoded using~$O(hn\log n)$ bits.

\begin{claim*}\label{claim:connected}
We may assume without loss of generality that~$G$ is connected.
\end{claim*}
\begin{proof}
Suppose that~$G$ is disconnected. Then we encode~$G$ by (1)~its component structure and (2)~a separate obstacle representation for each component. For~(1) we can simply store for each vertex the index of its component, which requires~$\log n$ bits per vertex and $n\log n$ bits overall. 
For~(2) we know by assumption that~$G$ admits an obstacle representation using at most~$h$ obstacles. So this is certainly true as well for every component~$C$ of~$G$ because it suffices to remove all vertices of~$V\setminus C$ from some representation for~$G$. 

Now suppose that we show that a connected graph on $n$ vertices can be represented using at most~$\alpha hn\log n$ bits, for some constant~$\alpha>0$. Let~$C_1,\ldots,C_k$ be the components of~$G$, and let~$n_i$ denote the size of~$C_i$, for~$1\le i\le k$. Then the obstacle representations of all components together can be represented using at most~$\sum_{i=1}^{n}\alpha hn_i\log n_i\le\alpha hn\log n$ bits. Altogether this amounts to at most~$n\log n+\alpha hn\log n\in O(hn\log n)$ bits to represent~$G$.

Note that the encoding described above does not represent an obstacle representation for~$G$ in general. But it uniquely represents~$G$, which suffices for the purposes of counting and proving our bound.
\end{proof}

For a set $S \subseteq\R^2$, we use $\partial S$ to denote the boundary of $S$.
We start with the following auxiliary result, whose statement is illustrated in Figure~\ref{fig-responsibilityApp}(a).

\begin{lemma}
\label{lem-boundaryIntersections}
Let $G$ be a connected graph with an obstacle representation $(D, \mathcal{O})$ and let $O \in \mathcal{O}$.
Let $u$ be a vertex of $G$ and let $z_0,z_1,z_2$ be three distinct points of $\partial O$.
Then there is an~$i \in \{0,1,2\}$ and a polygonal chain $\gamma^u_{z_{i-1},z_i,z_{i+1}} \subseteq \partial O$ between $z_{i-1}$ and $z_{i+1}$ such that $\gamma^u_{z_{i-1},z_i,z_{i+1}}$ intersects the ray $\overrightarrow{uz_i}$ at some point different from $z_i$ (the indices are taken modulo 3).
\end{lemma}
\begin{proof}
Suppose for contradiction that there is no such $i$ and $\gamma$ for some vertex $u$ and an obstacle $O$.
For every $i \in \{0,1,2\}$, the boundary of $O$ can be partitioned into polygonal chains $\gamma_i$ and $\gamma'_i$ that meet at common endpoints $z_{i-1}$ and $z_{i+1}$ and are otherwise disjoint as the obstacle $O$ is a simple polygon.
Moreover, $\gamma'_i$ contains $z_i$ while $\gamma_i$ does not.
By our assumption, no polygonal chain $\gamma_i$ thus intersects the ray $\overrightarrow{uz_i}$.
Then, however, the vertex $u$ is contained in the region $R$ bounded by the closed curve $\gamma_0 \cup \gamma_1 \cup \gamma_2$.
Since $O$ is a simple polygon, we have $R=O$.
This, however, contradicts the fact that $(D,\mathcal{O})$ is an obstacle representation of $G$ as $u$ is a vertex of $G$ lying inside $O$ and some edge of $G$ containing $u$ is blocked by $O$.
Note that we use the fact that $G$ is connected and thus $u$ is non-isolated in $G$.
\end{proof}

\begin{figure}
    \centering
    \includegraphics{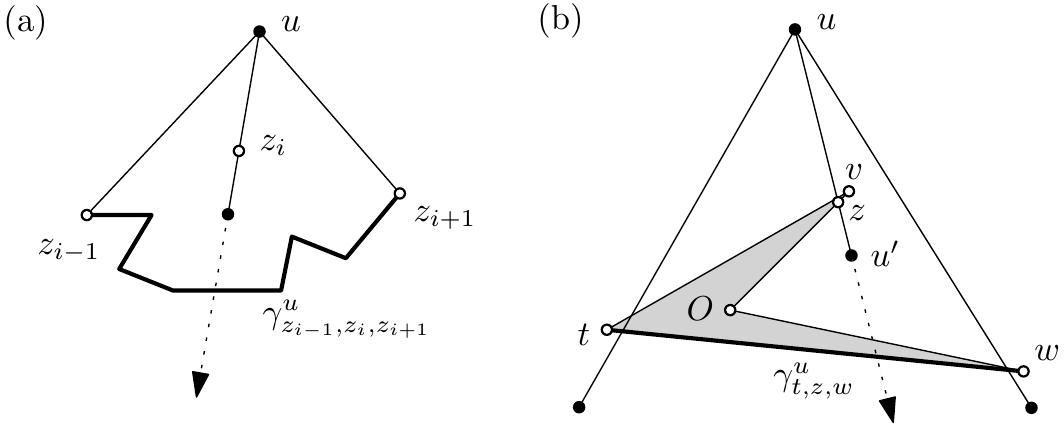}
    \caption{(a) An illustration of the statement of Lemma~\ref{lem-boundaryIntersections}. (b) An example where $u$ is responsible for $t$ and $w$ but not for $v$. The vertex $u'$ is responsible for $v$.}
    \label{fig-responsibilityApp}
\end{figure}

A vertex $v$ of a simple polygon $P$ is \emph{convex} if the internal angle of $P$ at $v$ is less than $\pi$; otherwise $v$ is a \emph{reflex} vertex of $P$.
Let $v$ be a convex vertex of an obstacle $O$ from a minimal obstacle representation $(D,\mathcal{O})$ of $G$ and let $e_1$ and $e_2$ be the two edges of $O$ adjacent to $v$.
Let $T$ be the triangle spanned by the edges $e_1$ and $e_2$.
Assume first that the interior of $T$ does not contain any point of $D$ nor a vertex of an obstacle from $\mathcal{O}$.
In this case, we call the vertex $v$ \emph{blocking}.

Since the representation $(D,\mathcal{O})$ is minimal, there is a line segment~$s$ representing a non-edge $e$ of $G$ such that $s$ intersects $e_1$ and $e_2$.
Otherwise, if $O$ has at least 4 vertices, we could remove $v$ and reduce the number of vertices of the obstacles from $\mathcal{O}$. 
If $s$ does not have any other intersections with obstacles from $\mathcal{O}$, then we say $e$ is \emph{forces} $v$. 
Then each blocking vertex of an obstacle with at least 4 vertices is forced by at least one non-edge of $G$. 
 
Let $u$ be an end-vertex of a non-edge $uu'$ that forces $v$ and let $z$ be an intersection point of $uu'$ with an edge of $\partial O$ incident to $v$.
We then say that $u$ is \emph{responsible} for $v$ if the following situation does not happen: there are non-edges of $G$ containing $u$ and forcing two other vertices $t$ and $w$ of $O$ that span the polygonal chain $\gamma^u_{t,z,w}$ from Lemma~\ref{lem-boundaryIntersections}; see  Figure~\ref{fig-responsibilityApp}(b).

We now prove a linear upper bound on the number of blocking vertices of an obstacle. 
To do so, we need the following auxiliary result.

\begin{lemma}
\label{lem-responsbility}
Let $G$ be a graph with a minimal obstacle representation $(D, \mathcal{O})$ and let $v$ be a blocking vertex of some obstacle $O \in \mathcal{O}$ with at least 4 vertices.
Then, there is at least one vertex of $G$ that is responsible for $v$.
\end{lemma}
\begin{proof}
Let $uu'$ be a non-edge of $G$ forcing $v$.
Let $z$ be an intersection point of $uu'$ with an edge of $\partial O$ incident to $v$. 
Suppose that $u$ is not responsible for $v$. 
Then, there are two other non-edges $ua$ and $ub$ of $G$ forcing two other vertices $t$ and $w$ of $O$, respectively, such that the polygonal chain $\gamma^u_{t,z,w} \subseteq \partial O$ between~$t$ and~$w$ intersects the ray~$\overrightarrow{uz}$ at some point~$x\ne z$; see Figure~\ref{fig-responsibility2App}(a) for an illustration.
Since $v$ is forced by $uu'$, the point $x$ lies behind $u'$ on the ray $\overrightarrow{uz}$.

We claim that~$u'$ is responsible for $v$. Suppose for a contradiction that this is not the case.
Then, analogously, there are non-edges $ua'$ and $ub'$ of $G$ forcing two other vertices $t'$ and $w'$ of $O$, respectively, such that the polygonal chain $\gamma^{u'}_{t',z,w'} \subseteq \partial O$ between $t'$ and $w'$ intersects the ray $\overrightarrow{u'z}$ at some point $x'\ne z$.
Similarly as before, the point $x'$ lies behind $u$ on the ray $\overrightarrow{u'z}$.

Since $t$ and $w$ are forced by $ua$ and $ub$, respectively, the vertices $a'$ and $b'$ lie outside the region $R$ bounded by the line segments $\overline{ut}$, $\overline{uw}$, and by the curve $\gamma^u_{t,z,w}$.
Otherwise at least one of the segments $ua$ and $ub$ has another intersection with $O$.
Analogously, the vertices $a$ and $b$ lie outside the region $R'$ bounded by the line segments $\overline{ut'}$, $\overline{uw'}$, and by the curve $\gamma^{u'}_{t',z,w'}$.
Since $z$ lies above $x$ on $\overrightarrow{uz}$ and above $x'$ on $\overrightarrow{u'z}$, we see $z \in R \cap R'$.
Thus, $z$ lies in the region bounded by the line segments $\overline{ua}$, $\overline{ub}$, $\overline{ua'}$, and $\overline{ub'}$.
However, the interior of at least one of these line segments has to be intersected by the portion of $\partial O$ between $z$ and, say, $t$.
This contradicts the fact that each of these four line segments is forcing a vertex of~$O$.
\end{proof}

\begin{figure}
    \centering
    \includegraphics{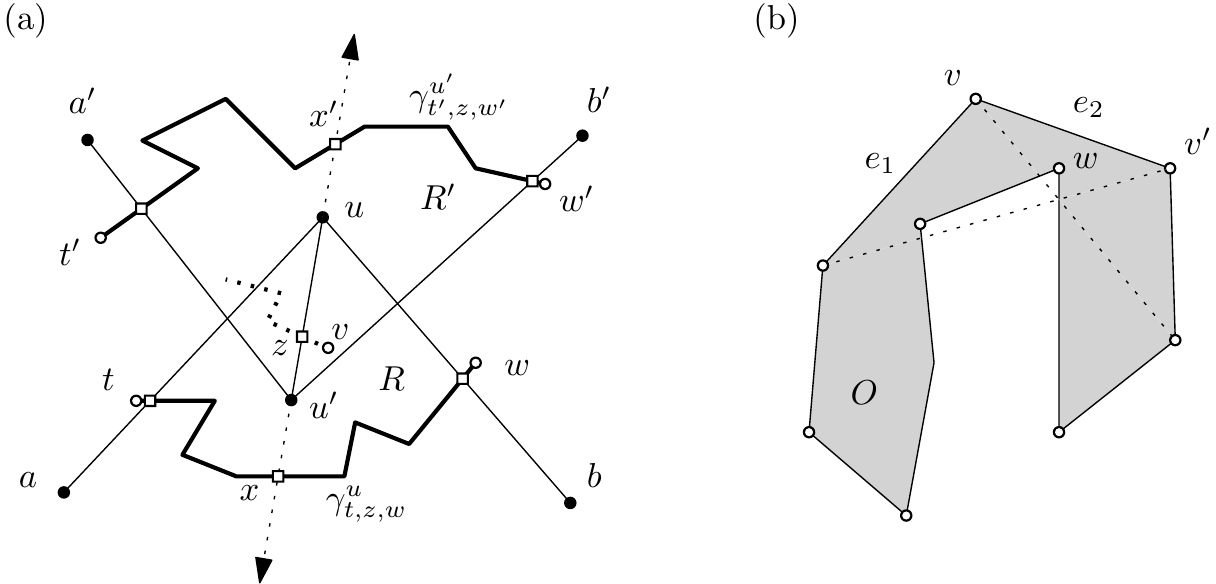}
    \caption{(a) An illustration of the proof of Lemma~\ref{lem-responsbility}. (b) The reflex vertex $w$ is responsible for the non-blocking convex vertices $v$ and $v'$ of the obstacle $O$.}
    \label{fig-responsibility2App}
\end{figure}

Now, we can use a double-counting argument to show that the number of convex vertices in each obstacle is at most linear in the number of vertices of $G$. 

\begin{lemma}
Let $O$ be an obstacle from a minimal obstacle representation $(D,\mathcal{O})$ of a connected graph $G$ with $n\ge 2$ vertices.
Then $O$ has at most $2n$ blocking vertices.
\label{lem-minimalObstaclesConvex}
\end{lemma}
\begin{proof}
First, we assume that $O$ has at least 4 vertices. 
Otherwise the statement is trivial as $n \geq 2$.
Thus, if a vertex $v$ of $O$ is blocking, the edges of $O$ adjacent to $v$ are intersected by a non-edge of $G$.

Now we show that every vertex of $G$ is responsible for at most two blocking vertices of~$O$. 
Suppose for contradiction that there is a vertex $u$ of $G$ that is responsible for at least three blocking vertices $v_0$, $v_1$, and $v_2$ of $O$.
For every $i \in \{0,1,2\}$, let $z_i$ be a point of $\partial O$ that lies on a non-edge of $G$ that forces $v_i$.
By Lemma~\ref{lem-boundaryIntersections}, there is $i \in \{0,1,2\}$ and a polygonal chain $\gamma^u_{z_{i-1},z_i,z_{i+1}} \subseteq \partial O$ between $z_{i-1}$ and $z_{i+1}$ such that $\gamma^u_{z_{i-1},z_i,z_{i+1}}$ intersects the ray $\overrightarrow{uz_i}$ at some point different from $z_i$ (the indices are taken modulo 3).
By the choice of $z_0,z_1,z_2$, the ray $\overrightarrow{uz_i}$ intersects the chain $\gamma^u_{v_{i-1},z_i,v_{i+1}}$ at some point different from $z_i$.
This contradicts the fact that $u$ is responsible for $v_i$, as we have excluded this situation in the definition of responsibility.

We use a double-counting argument on the number $x$ of pairs $(u,v)$, where $u$ is a vertex of $G$ that is responsible for a blocking vertex $v$ of $O$.
By Lemma~\ref{lem-responsbility}, for every blocking vertex $v$ of $O$, there is at least one vertex of $G$ that is responsible for $v$.
Thus, $x$ is at least as large as the number of blocking vertices of $O$.
On the other hand, since every vertex of~$G$ is responsible for at most two blocking vertices of $O$, we have $x \leq 2n$.
Altogether, the number of blocking vertices of $O$ is at most $2n$.
\end{proof}

We note that the proof of Lemma~\ref{lem-minimalObstaclesConvex} can be modified to show that each convex obstacle~$O$ has at most $n$ vertices.
This is because exactly two vertices of $G$ are then responsible for each vertex of $O$ and all vertices of $O$ are blocking.

Now, we consider the remaining convex vertices of obstacles from a minimal obstacle representation $(D,\mathcal{O})$ of $G$.
We call such vertices \emph{non-blocking}.

\begin{lemma}
Let $O$ be an obstacle from a minimal obstacle representation $(D,\mathcal{O})$ of a graph $G$ with $n$ vertices.
Then $O$ has at most $2r_O$ non-blocking vertices where $r_O$ is the number of reflex vertices of $O$.
\label{lem-notBlocking}
\end{lemma}
\begin{proof}
Without loss of generality, we assume that no three vertices of obstacles are collinear otherwise we apply a suitable perturbation.
Let $v$ be a non-blocking vertex of~$O$.
Let $T$ be the triangle spanned by the edges $e_1$ and $e_2$ of $O$ that are incident to $v$.
Since $v$ is non-blocking, the interior of $T$ contains either some vertex of $G$ or a vertex of an obstacle from $\mathcal{O}$; see Figure~\ref{fig-responsibility2App}(b).
Consider the line $\ell$ that contains $v$ and that is orthogonal to the axis of the acute angle between $e_1$ and $e_2$.
We sweep $\ell$ along the axis until, for the first time, the part of $\ell$ in the interior of $T$ meets a point $w$ that is not in the interior of $O$.
Then $w$ is a reflex vertex of $O$ since the obstacles are pairwise disjoint and no vertex of $G$ is contained in~$O$.
We say that $w$ is \emph{responsible} for $v$.

Now, we show that each reflex vertex of $O$ is responsible for at most two non-blocking vertices of $O$.
Let $v$ and $v'$ be two non-blocking vertices of $O$ and let $w$ be a reflex vertex of~$O$ responsible for $v$ and $v'$.
Let $T$ be the triangle spanned by the edges of $O$ that are incident to $v$.
The only points from $\partial O$ that are in the swept part of $T$ just before we meet $w$ are from $e_1$ and $e_2$.
This claim is also true for the triangle $T'$ that is analogously defined for $v'$.
Thus, since no three vertices of obstacles are collinear, the triangles $T$ and $T'$ share an edge.
Therefore, $w$ cannot be responsible for more than two such vertices.

We now use a double-counting argument on the number $x$ of pairs $(w,v)$, where $w$ is a reflex vertex of $O$ that is responsible for a non-blocking vertex $v$ of $O$.
For every such vertex $v$ of $O$, there is at least one reflex vertex of $O$ that is responsible for $v$.
Thus, $x$ is at least as large as the number of non-blocking vertices of $O$.
On the other hand, since every reflex vertex of $O$ is responsible for at most two non-blocking vertices of $O$, we have $x \leq 2r_O$.
Altogether, the number of non-blocking vertices of $O$ is at most $2r_O$.
\end{proof}

The only remaining vertices of obstacles are the reflex vertices.
The following result gives an estimate on their number in a minimal obstacle representation.

\begin{lemma}
Let $(D,\mathcal{O})$ be a minimal obstacle representation of a  connected graph $G$ with $n$ vertices.
Then the total number of reflex vertices of obstacles from $\mathcal{O}$ is at most $2n+5$.
\label{lem-minimalObstaclesReflex}
\end{lemma}
\begin{proof}
Without loss of generality, we assume that no three vertices of obstacles are collinear otherwise we apply a suitable perturbation.
Let $v$ be a reflex vertex of an obstacle $O$ from $\mathcal{O}$.
We use $e_1$ and $e_2$ to denote the two edges of $O$ that contain $v$ and $T$ to denote the triangle spanned by $e_1$ and $e_2$.
The interior of $T$ contains some point from $D \cup \bigcup \mathcal{O}$ as otherwise we could remove $v$ by considering a new obstacle $O \cup T$, which is impossible as $(D,\mathcal{O})$ is minimal.

Similarly as in the proof of Lemma~\ref{lem-notBlocking}, we consider the line $\ell$ that contains $v$ and that is orthogonal to the axis of the acute angle between $e_1$ and $e_2$.
We sweep $\ell$ along the axis until the part of $\ell$ in the interior of $T$ meets a point $w$ from $D \cup \bigcup \mathcal{O}$ for the first time.
Then $w$ is either some vertex of $D$ or a vertex of some obstacle from $\mathcal{O}$.
We distinguish three cases.

First, suppose that $w$ is a vertex of another obstacle $O'$ from $\mathcal{O}$.
Then we can merge the obstacles $O$ and $O'$ by adding a small polygon inside the a small neighborhood of the swept portion of $T$ containing $w$; see Figure~\ref{fig-reflexApp}(a).
This, however, is impossible as $(D,\mathcal{O})$ is minimal and, in particular, uses the minimum number of obstacles for $G$.

\begin{figure}
    \centering
    \includegraphics{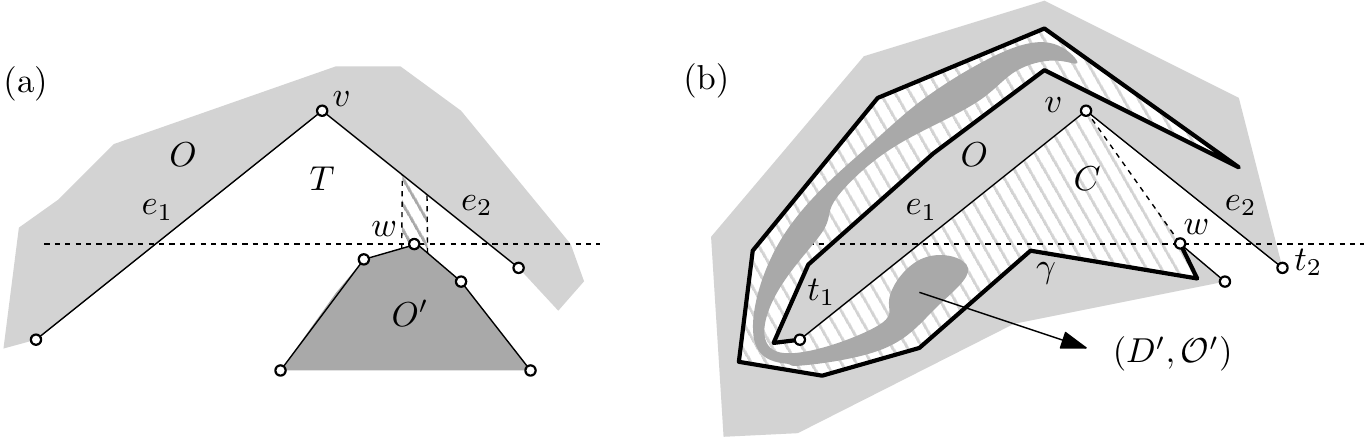}
    \caption{An illustration of the proof of Lemma~\ref{lem-minimalObstaclesReflex}. (a) Merging two obstacles. (b) If $w \in O$, then there is an isolated component $(D'\mathcal{O}')$ in the region $C$ (denoted by stripes).}
    \label{fig-reflexApp}
\end{figure}

Second, assume that $w$ is a vertex of the obstacle $O$. 
Let $t_1$ and $t_2$ be the endpoints of $e_1$ and $e_2$, respectively, different from $v$.
Since $O$ is a simple polygon, there is a polygonal chain $\gamma \subseteq \partial O$ between $w$ and some $t_i$ enclosing together with line segments $e_i$ and $\overline{vw}$ a region $C$ such that $C$ does not contain $t_{3-i}$; see Figure~\ref{fig-reflexApp}(b).
By symmetry, we assume $i=1$.
Let $D' = D \cap C$ and $\mathcal{O}'$ be the  set of obstacles from $\mathcal{O}$ that are contained in~$C$.
Note that $(D',\mathcal{O}') = (D,\mathcal{O}\setminus \{O\})$ as $G$ is connected and $(D',\mathcal{O}')$ is an obstacle representation of a subgraph of $G$ that is isolated from everything outside of $C$, since every line segment connecting a point of $D'$ with a point from $D \setminus D'$ intersects $O$.
We say that that the edge $vw$ is a \emph{cutting segment} \emph{bounding} $C$. 

We show that there are at most five reflex vertices that are endpoints of cutting segments.
Suppose for contradiction that we have cutting segments $v_1w_1,\dots,v_mw_m$ bounding regions $C_1,\dots,C_m$, respectively, where $m \geq 6$ and $v_1,\dots,v_m$ are distinct.
Then all the cutting segments have endpoints in the same obstacle $O$,
for every $i\in [m]$, as the drawing $D$ has $O_i$ in the outerface,
which cannot contain any other obstacle as otherwise we could merge
them, contradicting the minimality of $(D,\mathcal{O})$.
The line segments $v_1w_1,\dots,v_mw_m$ do not cross as, after possible relabeling, we have $C_1 \supset \dots \supset C_m$ and $v_iw_i$ is contained in $\partial C_i$ for every $i \in [m]$; see Figure~\ref{fig-reducingObstacle}.
For each $i \in [m-1]$, let $C'_i = C_i \setminus C_{i+1}$.
It follows from the definition of cutting segments that the regions $C'_1,\dots,C'_{m-1}$ do not contain any point of $D$ nor a point of an obstacle from $\mathcal{O}$ in their interior.
We now reduce the number of vertices of $O$.
The set $O \cup C'_1 \cup \dots \cup C'_{m-1}$ is a polygon with a single hole corresponding to the region $C_m$.
Let $t$ be a point of $C_m$ with the smallest $y$-coordinate and note that we can assume that $t$ is a vertex of $O$.
Let $s \subseteq O$ be a vertical line segment with the topmost point $t$.
By the choice of $t$, the other endpoint of $s$ is a point from $\partial O$ that does not lie in $\partial C_m$.
Thus, we can replace $O$ with a new obstacle $O'$ obtained from $O \cup C'_1 \cup \dots \cup C'_{m-1}$ by removing a small rectangular neighborhood of $s$; see Figure~\ref{fig-reducingObstacle}.
Since $t$ is a vertex of $O$, we add at most three new vertices to $O'$ with respect to~$O$.
On the other hand, we remove at least $m-2 \geq 4$ vertices, as $v_2,\dots,v_{m-1}$ will no longer be vertices of $O'$.
This contradicts the minimality of $(D,\mathcal{O})$.

\begin{figure}
    \centering
    \includegraphics{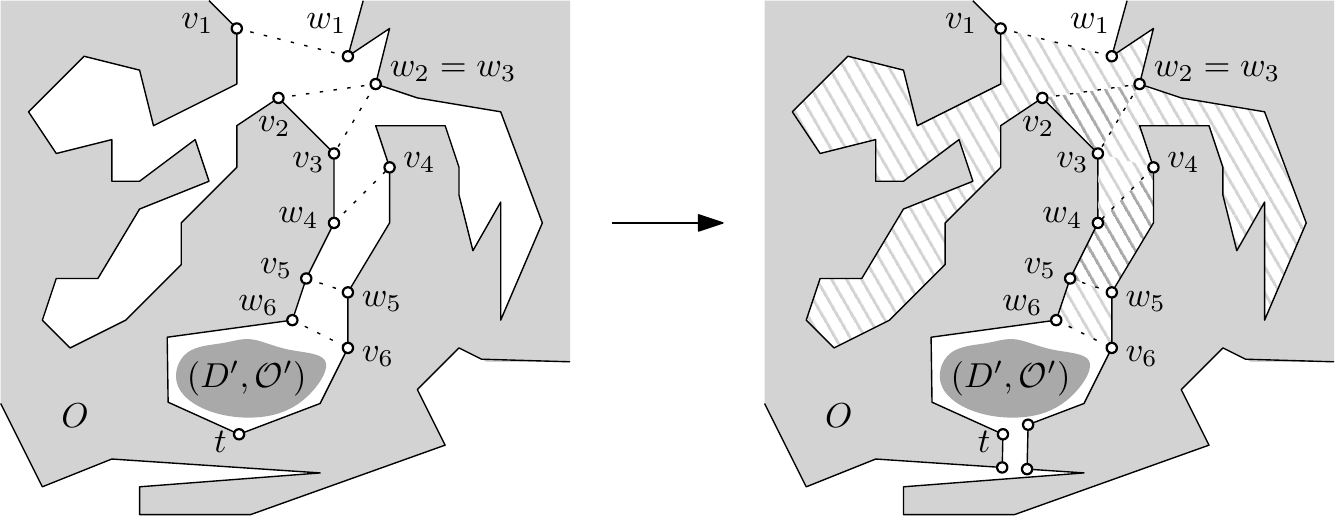}
    \caption{An illustration of the proof of Lemma~\ref{lem-minimalObstaclesReflex}. Reducing the number of vertices of an obstacle $O$ by filling up regions $C'_1,\dots,C'_5$ and ``drilling'' a vertical tunnel from the vertex $t$ of $C_6$.}
    \label{fig-reducingObstacle}
\end{figure}

Finally, assume that $w$ is a vertex of $G$.
Then we say that $w$ is \emph{responsible} for $v$.
We now show that every vertex of $G$ is responsible for at most two reflex vertices of an obstacle.
This is done via a similar argument as in the proof of Lemma~\ref{lem-notBlocking}.
Let $v$ and $v'$ be two reflex vertices of some obstacles $O$ and $O'$, respectively, and let $w$ be a vertex of~$G$ responsible for $v$ and $v'$.
Let $T$ be the triangle spanned by the edges $e_1$ and $e_2$ of $O$ that are incident to $v$.
The only points from any obstacle that are in the swept part of $T$ just before we meet $w$ are from $e_1$ and $e_2$.
This claim is true for the triangle $T'$ that is defined analogously for $v'$.
Thus, since no three vertices of obstacles are collinear, the triangles $T$ and $T'$ share an edge.
Therefore, $w$ cannot be responsible for more than two reflex vertices.

Now, every vertex of $G$ is responsible for at most two reflex vertices in total.
Thus, since apart for at most 5 endpoints of cutting segments for each reflex vertex $v$ of an obstacle there is at least one vertex of $G$ 
that is responsible for~$v$,
we get that the total number of reflex vertices is at most $2n+5$.
\end{proof}

To bound the number $g(h,n)$ of connected $n$-vertex graphs with obstacle number at most $h$, we follow the approach by Mukkamala, Pach, and P{\'a}lv{\"o}lgyi~\cite{mpp-lbong-12}, which is based on encoding point sets by order types.
Let $P$ and $Q$ be two sets of points in the plane in general position.
Then, $P$ and $Q$ \emph{have the same order type} if there is a one-to-one correspondence between them with the property that the orientation of any triple of points from $P$ is the same as the orientation of the corresponding triple of points from $Q$.
Recall that we consider only point sets in general position so the orientation is either clockwise or counterclockwise.
So, formally, an \emph{order type} of a set $P$ of points in the plane in general position is  a mapping $\chi$ that assigns $\chi(a, b, c) \in \{+1, -1\}$ to each ordered triple $(a, b, c)$ of points from $P$, indicating whether $a, b, c$ make a left turn (+1) or a right turn (-1).

We use the following classical estimates by Goodman and Pollack~\cite{gp-upcp-86,goPo93} on the number of order types of sets of $n$ points.

\begin{theorem}[\cite{gp-upcp-86,goPo93}]
\label{thm-orderType}
The number of order types on $n$ labeled points in $\R^2$ is at most
$\left(\frac{n}{2}\right)^{4(1+o(1))n} \in 2^{O(n \log n)}$.
\end{theorem}

Now, for each connected $n$-vertex graph $G$ with obstacle number at most $h$, fix a minimal obstacle representation $(D,\mathcal{O})$ of $G$ with $|\mathcal{O}| \leq h$.
Let $s$ be the number of vertices of obstacles from $\mathcal{O}$ and let $r_O$ be the number of reflex vertices of an obstacle $O$ from $\mathcal{O}$.
We use $r$ to denote the total number of reflex vertices $\sum_{O \in \mathcal{O}} r_O$.
Since every vertex of an obstacle is either convex (blocking or non-blocking) or reflex, Lemmas~\ref{lem-minimalObstaclesConvex},~\ref{lem-notBlocking}, and~\ref{lem-minimalObstaclesReflex} give 
\[s \leq r + \sum_{O \in \mathcal{O}} (2r_O + 2n) \leq 3r + |\mathcal{O}|\cdot 2n \leq 2hn+6n+15.\]

Let $S$ be a sequence of $n+s$ labeled points starting with the vertices of $G$ and followed by the vertices of the obstacles (listed one by one, in cyclic order, and properly separated from one another).
Then, the order type of $S$ determines the graph $G$ as it stores all the information about incidences in the obstacle representation $(D,\mathcal{O})$. 
By Theorem~\ref{thm-orderType}, the number of such order types is at most $2^{\alpha(n+s) \log{(n+s)}}$ for some suitable constant $\alpha>0$.
Since $s \leq 2hn+6n+15$ and $h \leq \binom{n}{2}$, we obtain the  desired upper bound
\[g(h,n) \in 2^{\alpha(n+2hn+6n+15) \log{(n+2hn+6n+15)}} \subseteq 2^{O(hn\log{n})}.\]

\subsection{Proof of Theorem~\ref{thm-complexityConvex}}
\label{sec-complexityConvex}

We show that the number of graphs on $n$ vertices that have   convex obstacle number at most~$h$ is bounded from above by~$2^{O(n(h+\log n))}$. 
That is, we prove the bound $\fc(h,n) \in 2^{O(n(h+\log n))}$, which improves the earlier estimate $\fc(h,n) \in 2^{O(hn\log{n})}$ 
by Mukkamala, Pach, and P{\'a}lv{\"o}lgyi~\cite{mpp-lbong-12} and for $h<n$ asymptotically matches the lower bounds by Balko, Cibulka, and Valtr~\cite{bcv-dgusno-DCG18}.

To this end, we find an efficient encoding of an obstacle representation $(D,\mathcal{O})$ for any $n$-vertex graph $G=(V,E)$ that uses at most $h$ convex obstacles.
We identify the vertex set $V$ and the points of $D$ that represent the vertices from $V$.
The first part of the encoding is formed by the order type of $V$.
By Theorem~\ref{thm-orderType}, there are at most $2^{O(n\log{n})}$ order types of sets of $n$ points in the plane in general position.

It remains to encode the obstacles and their interaction with the line segments between points from $V$.
To do that, we use so-called radial systems.
The \emph{clockwise radial system} $R$ of~$V$ assigns to each vertex $v\in V$ the clockwise cyclic ordering $R(v)$ of the~$n-1$ rays in~$D$ that start from~$v$ and pass through a vertex from~$V\setminus\{v\}$.  
The order type of~$V$ uniquely determines the radial system~$R$ of~$V$. 
Essentially, this also holds in the other direction: There are at most~$n-1$ order types that are compatible with a given radial system~\cite{acklv-rpsotro-16}.

For a vertex $v \in V$ and an obstacle $O \in \mathcal{O}$, let $I_O(v)$ be the subsequence of rays in~$R(v)$ that
intersect $O$.
We say that a subset $I$ of $R(v)$ is an \emph{interval} 
if there are no four consecutive elements $a,b,c,d$ in the radial order $R(v)$ such that $a,c \in I$ and $b,d \notin I$.
Since $O$ is connected, we get the following result.

\begin{observation}
\label{obs-blockingInterval}
For every vertex $v \in V$ and each obstacle $O \in \mathcal{O}$, the set $I_O(v)$ forms an interval in the radial ordering $R(v)$. \qed
\end{observation}

Thus, we call the set $I_O(v)$ the \emph{blocking interval} of the pair~$(v,O)$.
For convex obstacles, knowing which rays from a vertex towards other vertices are blocked by an obstacle suffices to know which edges are blocked by this obstacle.

\begin{lemma}
\label{lem:blockingInterval}
For two vertices $u,v \in V$, the pair $\{u,v\}$ is a non-edge of $G$ if and only if there is an obstacle $O \in \mathcal{O}$ such that $u \in I_O(v)$ and $v \in I_O(u)$.
\end{lemma}
\begin{proof}
Consider two vertices $u,v \in V$.
Let $R_u$ be the ray starting at $u$ and passing through $v$ and let $R_v$ be the ray starting at $v$ and passing through $u$.
Then $uv = R_u \cap R_v$.
Since each obstacle $O \in \mathcal{O}$ is convex, the intersection of $O$ with each ray or line segment is a (possibly empty) line segment.
Thus, $O$ intersects both rays $R_u$ and $R_v$ if and only if $O$ intersects the line segment $\overline{uv}$.
Since $(D,\mathcal{O})$ is an obstacle representation of $G$, we then have $\{u,v\} \notin E$ if and only if $u \in I_O(v)$ and $v \in I_O(u)$.
\end{proof}

We note that Lemma~\ref{lem:blockingInterval} is not true if the obstacles from $\mathcal{O}$ are not convex, since both rays $R_u$ and $R_v$ could then be blocked although the line segment~$\overline{uv}$ is not. 

By Observation~\ref{obs-blockingInterval} and Lemma~\ref{lem:blockingInterval}, it suffices to encode the blocking intervals $I_O(v)$ for every $v \in V$ and $O \in \mathcal{O}$.
In the following we will describe how to obtain an encoding of size~$2^{O(hn)}$ for the blocking intervals, which together with the~order type of $V$ yields the claimed bound on $f_c(h,n)$.
In order to describe our approach, it is more convenient to move to the dual setting, using the standard \emph{projective duality transformation} that maps a point~$p=(p_x,p_y) \in \R^2$ to the line~$p^* = \{(x,y) \in \R^2 \colon y=p_xx-p_y\}$, and that maps a non-vertical line~$\ell = \{(x,y) \in \R^2 \colon y=mx+b\}$ to the point~$\ell^*=(m,-b)$.
This map is an involution, that is, $(p^*)^*=p$ and~$(\ell^*)^*=\ell$.
Moreover, it preserves incidences, that is, $p\in\ell\Longleftrightarrow\ell^*\in p^*$, and is order-preserving as a point $p$ lies above~$\ell\Longleftrightarrow\ell^*$ lies above~$p^*$. 

So consider the arrangement~$\mathcal{A}$ of the set~$V^*$ of~$n$ lines dual to the points in~$V$.
Note that the combinatorial structure of~$\mathcal{A}$ itself, that is, the sequences of intersections with other lines along each line, can be obtained from the radial system of~$V$, cf.~\cite{w-cvgnl-85}. (To be able to identify the vertical direction and the $x$-order of the vertices, we add a special point very high above all other points.)  
Let $O$ be an obstacle from~$\mathcal{O}$.
Define a map~$\tau$ that assigns to each~$x\in\R$ the \emph{upper tangent} of slope~$x$ to~$O$, that is, the line~$\tau(x)$ of slope~$x$ that is tangent to~$O$ and such that~$O$ lies below~$\tau(x)$; see Figure~\ref{fig-dualMappingApp} for an illustration.
Now, consider the dual~$\tau^*$ of~$\tau$, defined by~$\tau^*(x)=(\tau(x))^*$.
Note that, by definition, every line in~the image of $\R$ via $\tau$ passes through a vertex of the upper envelope~$O^+$ of the convex hull of~$O$.
Consequently, each point~$\tau^*(x)$ lies on a line that is dual to a vertex of~$O^+$.
In other words, $\tau^*$ is a piecewise linear function that is composed of line segments along the lines dual to the vertices of~$O^+$.
The order of these line segments from left to right (in increasing $x$-order) corresponds to primal tangents of increasing slope and, therefore, to the order of the corresponding vertices of~$O^+$ from right to left (in decreasing $x$-order).
As primal $x$-coordinates correspond to dual slopes, the slopes of the line segments along~$\tau^*$ monotonically decrease from left to right, and so~$\tau^*$ is concave.

\begin{figure}
    \centering
    \includegraphics{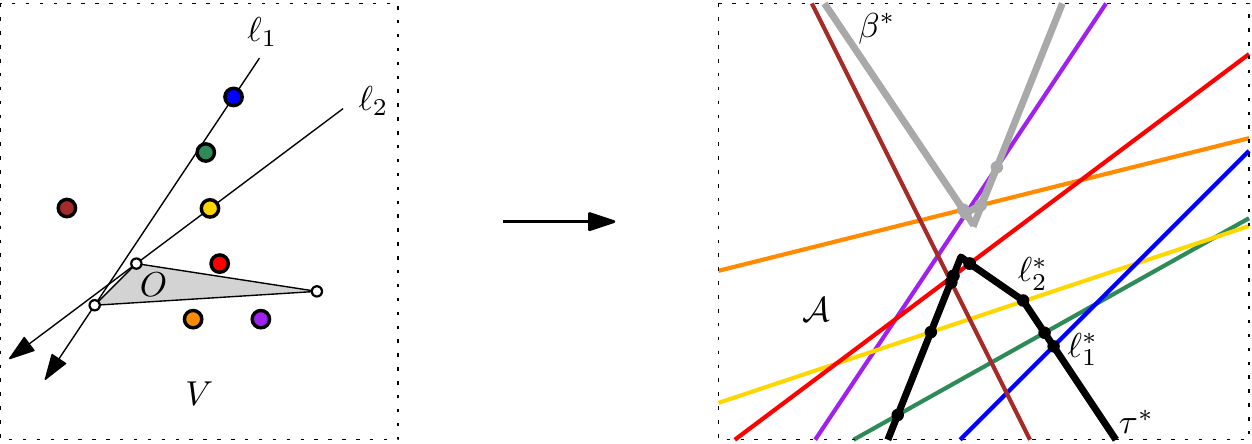}
    \caption{An example of a vertex set $V$ with its dual line arrangement $\mathcal{A}= V^*$. Points and lines of the same color are dual to each other. The curve $\tau^*$ for the upper envelope $O^+$ of the obstacle $O$ is denoted black and the curve $\beta^*$ for the lower envelope $O^-$ of the obstacle $O$ is denoted gray.}
    \label{fig-dualMappingApp}
\end{figure}

The primal line of a point in~$v^*\cap\tau^*$, for some~$v\in V$, passes through~$v$ and is an upper tangent to~$O$.
So such an intersection corresponds to an endpoint of the blocking interval $I_O(v)$.
In order to obtain all endpoints of the blocking interval $I_O(v)$, we eventually also consider the \emph{lower tangents} to~$O$ in an analogous manner.
The corresponding function~$\beta^*$ is convex and consists of segments along lines dual to the vertices of the lower convex hull of~$O$, in increasing $x$-order.

It remains to describe how to compactly encode the intersections of~$\tau^*$ with~$\mathcal{A}$.
To do so, we apply a result by Knuth~\cite{k-ah-92}; see also the description by Felsner and Valtr~\cite{fv-ccap-11}.
A set $\mathcal{A'}$ of biinfinite curves in~$\R^2$ forms an \emph{arrangement of pseudolines} if each
pair of curves from~$\mathcal{A}'$ intersects in a unique point, which corresponds to a proper, transversal crossing.
A curve~$\alpha$ is a \emph{pseudoline with respect to~$\mathcal{A'}$} if~$\alpha$ intersects each curve in~$\mathcal{A'}$ at most once.
Let $\gamma$ be a curve 
that intersects each curve from $\mathcal{A}'$ in a finite number of points.
The \emph{cutpath of $\gamma$ in $\mathcal{A'}$} is the sequence of intersections of $\gamma$ with the pseudolines from $\mathcal{A'}$ along~$\gamma$. 

\begin{theorem}[\cite{k-ah-92}]
\label{thm:pseudoline}
If~$\alpha$ is a pseudoline with respect to an arrangement~$\mathcal{A'}$ of $n$ pseudolines, then there are at most~$O(3^n)$ cutpaths of $\alpha$ in~$\mathcal{A'}$.
\end{theorem}

Using Theorem~\ref{thm:pseudoline}, we can now estimate the number of cutpaths  of $\tau^*$ with respect to $\mathcal{A}$.

\begin{lemma}
\label{lem:cutpaths}
There are~$O(216^n)$ cutpaths of $\tau^*$ in~$\mathcal{A}$.
\end{lemma}

\begin{proof}
If~$\tau^*$ was a pseudoline with respect to~$\mathcal{A}$, then using Theorem~\ref{thm:pseudoline} we could conclude that there are only~$O(3^n)$ cutpaths of $\tau^*$ in~$\mathcal{A}$.
However, $\tau^*$ is not a pseudoline with respect to~$\mathcal{A}$ in general because a vertex from $V$ can be contained in two tangents to~$O^+$, in which case~$\tau^*$ intersects the corresponding line of~$\mathcal{A}$ twice.
More precisely, this is the case exactly for those points from~$V$ that lie vertically above~$O$, whereas all points that lie to the left or to the right of~$O$ have only one tangent to~$O^+$ and one tangent to the lower envelope~$O^-$ of the convex hull of~$O$.
All points that lie vertically below~$O$ have no tangent to~$O^+$ but two tangents to~$O^-$. (We assign the tangents to the leftmost and rightmost vertex of~$O$ as we see fit.)

To remedy this situation, we split the lines of~$\mathcal{A}$ into two groups: Let~$\mathcal{A}_1$ denote the arrangement of those lines from~$\mathcal{A}$ that intersect~$\tau^*$ at most once, and let~$\mathcal{A}_2$ denote the arrangement of those lines from~$\mathcal{A}$ that intersect~$\tau^*$ exactly twice. To describe a cutpath of~$\tau^*$ in~$\mathcal{A}$ it suffices to encode (1)~which lines are in~$\mathcal{A}_1$; (2)~the cutpath of~$\tau^*$ in~$\mathcal{A}_1$; (3)~the cutpath of~$\tau^*$ in~$\mathcal{A}_2$; and (4)~a bitstring that tells us whether the next crossing is with a line from~$\mathcal{A}_1$ or from~$\mathcal{A}_2$ when walking along~$\tau^*$ from left to right. For~(1) there are~$2^n$ options, for~(2) there are $O(3^n)$~options by Theorem~\ref{thm:pseudoline}, and for~(4) there are at most~$2^{2n}=4^n$ options. It remains to argue how to encode~(3).

To encode the cutpath of~$\tau^*$ in~$\mathcal{A}_2$, we split each line~$\ell\in\mathcal{A}_2$ at its leftmost crossing~$\mathrm{c}_\ell$ with~$\tau^*$ and construct a pseudoline~$\ell'$ by taking the part of~$\ell$ to the left of~$\mathrm{c}_\ell$ and extending it to the right by an almost vertical downward ray (that is, so that no vertex of~$\mathcal{A}$ lies between the vertical downward ray from~$\mathrm{c}_\ell$ and~$\ell'$). Now the collection~$\{\ell'\colon\ell\in\mathcal{A}_2\}$ induces a pseudoline arrangement~$\mathcal{A}_2'$ such that~$\tau^*$ is a pseudoline with respect to~$\mathcal{A}_2'$. Note that~$\tau^*$ crosses all pseudolines in~$\mathcal{A}_2'$ from below. We claim that the cutpath of~$\tau^*$ in~$\mathcal{A}_2'$ together with~$\mathcal{A}_2$ suffices to reconstruct~$\mathcal{A}_2'$.

Let us prove this claim. So suppose we are given~$\mathcal{A}_2$ and the cutpath~$p$ of~$\tau^*$ in~$\mathcal{A}_2'$. The cutpath~$p$ is encoded as a bitstring that determines if the path continues by leaving the current cell through its leftmost or rightmost edge and the set of lines that is crossed as a middle (i.e., neither leftmost nor rightmost) edge of a cell as soon as it is encountered in such a way. This information suffices to decribe~$p$ because it can be shown that~$p$ encounters every line of~$\mathcal{A}_2'$ at most once as a middle edge of a cell, cf.~\cite{fv-ccap-11,k-ah-92}. 

To recover~$\mathcal{A}_2'$ we trace~$p$ from left to right. We know that the starting cell of~$p$ in~$\mathcal{A}_2$ is the bottom cell of~$\mathcal{A}_2$. Then we just proceed by pretending that~$p$ was described with respect to~$\mathcal{A}_2$. However, at each crossing we cut the line~$\ell$ crossed by~$p$ and discard the part of~$\ell$ to the right of the crossing and replace it by an almost vertical downward ray, as described above in the construction of~$\mathcal{A}_2'$. Hence, as~$p$ crosses only lines that approach it from above, no line whose leftmost crossing with~$p$ has already been processed is ever considered again. Moreover, as we trace~$p$ from left to right, whenever a crossing of a line~$\ell$ with~$p$ is processed, it is the leftmost crossing of~$\ell$ with~$p$. In other words, all crossings that are discovered during the trace correspond to crossings of~$p$ in~$\mathcal{A}_2'$. 

Finally, we claim that during the trace~$p$ encounters every line at most once as a middle edge of a cell, and so the middle set information (which has been created with respect to~$\mathcal{A}_2'$, but we interpret it in our work in progress arrangement, which is a combination of~$\mathcal{A}_2$ and~$\mathcal{A}_2'$) suffices for the reconstruction. To see this, suppose that during our trace of~$p$ a line~$\ell$ is encountered as a middle edge of a cell~$C$ and then later again as a middle edge of a cell~$C'$. Let~$\ell^-$ be the line that contributes the leftmost edge to the upper boundary of~$C'$. Then the unique crossing of~$\ell^-$ and~$\ell$ lies to the right of or on~$\ell'\cap C'$ along~$\ell'$ and~$\ell'$ lies below~$\ell$ to the left of~$C'$. Further, we know that~$\ell^-$ does not cross the part of~$p$ traced so far (i.e., to the left of~$C'$) because if it did, we would have cut~$\ell'$ at this crossing. It follows that to the left of~$C'$ the pseudolines~$p$ and~$\ell$ are separated by~$\ell'$ and so~$\ell$ cannot appear on the boundary of any cell~$C$ traced by~$p$ so far, a contradiction. So, as claimed, during the trace~$p$ encounters every line at most once as a middle edge of a cell. This also completes the proof of our claim that the cutpath of~$\tau^*$ in~$\mathcal{A}_2'$ together with~$\mathcal{A}_2$ suffices to reconstruct~$\mathcal{A}_2'$.

In a symmetric fashion (vertical flip of the plane) we obtain an arrangement~$\mathcal{A}_2''$ and cutpath to describe the rightmost crossings of the lines in~$\mathcal{A}_2$ with~$\tau^*$. Combining both we can reconstruct the cutpath of~$\tau^*$ in~$\mathcal{A}_2$. By Theorem~\ref{thm:pseudoline} there are~$O(3^n)$ options for each of the two cutpaths. So there are~$O(9^n)$ options for~(3) and altogether~$O(216^n)$ cutpaths of~$\tau^*$ in~$\mathcal{A}$.
\end{proof}

\begin{figure}
    \centering
    \includegraphics{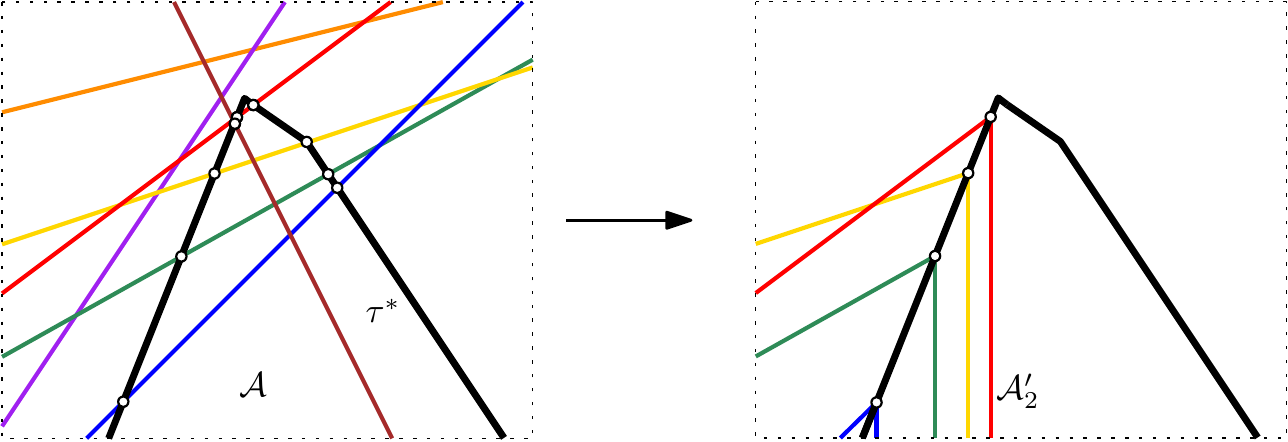}
    \caption{Cutting the lines of $\mathcal{A}$ that cross~$\tau^*$ twice at their leftmost crossing with~$\tau^*$. Then~$\tau^*$ is a pseudoline with respect to the resulting pseudoline arrangement $\mathcal{A}_2'$.}
    \label{fig-cutpathsApp}
\end{figure}

To summarize, we encode the obstacle representation $(D,\mathcal{O})$ of $G$ by first encoding the order type of $V$.
By Theorem~\ref{thm-orderType}, there are at most $2^{O(n\log{n})}$ choices for the order type of a set of $n$ points.
The order type of $V$ determines the arrangement $\mathcal{A}$ of $n$ lines that are dual to the points from $V$.
By Observation~\ref{obs-blockingInterval} and Lemma~\ref{lem:blockingInterval}, it suffices to encode the endpoints of the blocking intervals $I_O(v)$ for every $v \in V$ and $O \in \mathcal{O}$, which are defined using the radial system of $V$ that is also determined by the order type of~$V$.
For each obstacle $O \in \mathcal{O}$, the endpoints of all the intervals $I_V(O)$ are determined by the cutpath of the curve $\tau^*$ in $\mathcal{A}$ constructed for the upper tangents of $O$ and the analogous curve $\beta^*$ constructed for the lower tangents of $O$.
By Lemma~\ref{lem:cutpaths}, there are at most $O(216^n)$ cutpaths of $\tau^*$ in $\mathcal{A}$ and the same estimate holds for $\beta^*$.
This gives at most $O((216^n)^2)$ possible ways how to encode a single obstacle.
Altogether, we thus obtain
\[\fc(h,n) \in 2^{O(n\log{n})} \cdot O(216^{2hn}) \subseteq 2^{O(n\log{n} + hn)},\] which concludes the proof of Theorem~\ref{thm-complexityConvex}.

\begin{rem}
It is natural to wonder whether we cannot use Theorem~\ref{thm-complexityConvex} to obtain a linear lower bound on $\obs(n)$ by splitting each obstacle in an obstacle representation $(D,\mathcal{O})$ of a graph $G$ with $|\mathcal{O}| = h$ into few convex pieces and apply Theorem~\ref{thm-complexityConvex} on the resulting obstacle representation $(D,\mathcal{O}')$ of $G$ with convex obstacles.
But an obstacle $O$ with $r_O$ reflex vertices may require splitting into $\Omega(r_O)$ convex pieces and therefore we may have $|\mathcal{O}'| \in \Omega(h+r)$ where $r$ is the total number of reflex vertices of obstacles in $\mathcal{O}$.
Unfortunately, $r$ can be linear in $n$ and then the resulting bound $2^{O(n\log{n} + (h+r)n)}$ may be larger then the number $2^{\binom{n}{2}}$ of $n$-vertex graphs, in which case the proof of Theorem~\ref{thm-lowerBoundConvex} from Section~\ref{sec-lowerBound} no longer works.
\end{rem}

\section{A Lower Bound on the Obstacle Number of Drawings}
\label{sec-drawing}

Here, we prove Theorem~\ref{thm-drawing} by showing that there is a constant $\delta>0$ such that, for every $n$, there exists a graph $G$ on $n$ vertices and a drawing $D$ of $G$ such that ${\rm obs}(D) \geq \delta \cdot n^2$.
The proof is constructive. We first assume that $n$ is even.

A set of points in the plane is a \emph{cup} if all its points lie on the graph of a convex function.
Similarly, a set of points is a \emph{cap} if all its points lie on the graph of a concave function.
For an integer $k \geq 3$, a convex polygon $P$ is \emph{$k$-cap free} if no $k$ vertices of $P$ form a cap of size $k$. Note that $P$ is $k$-cap free if and only if it is bounded from above by at most $k-2$ segments (edges of $P$).
We use $e(P)$ to denote the leftmost edge bounding $P$ from above; see Figure~\ref{fig-drawing}(a).

\begin{figure}
    \centering
    \includegraphics{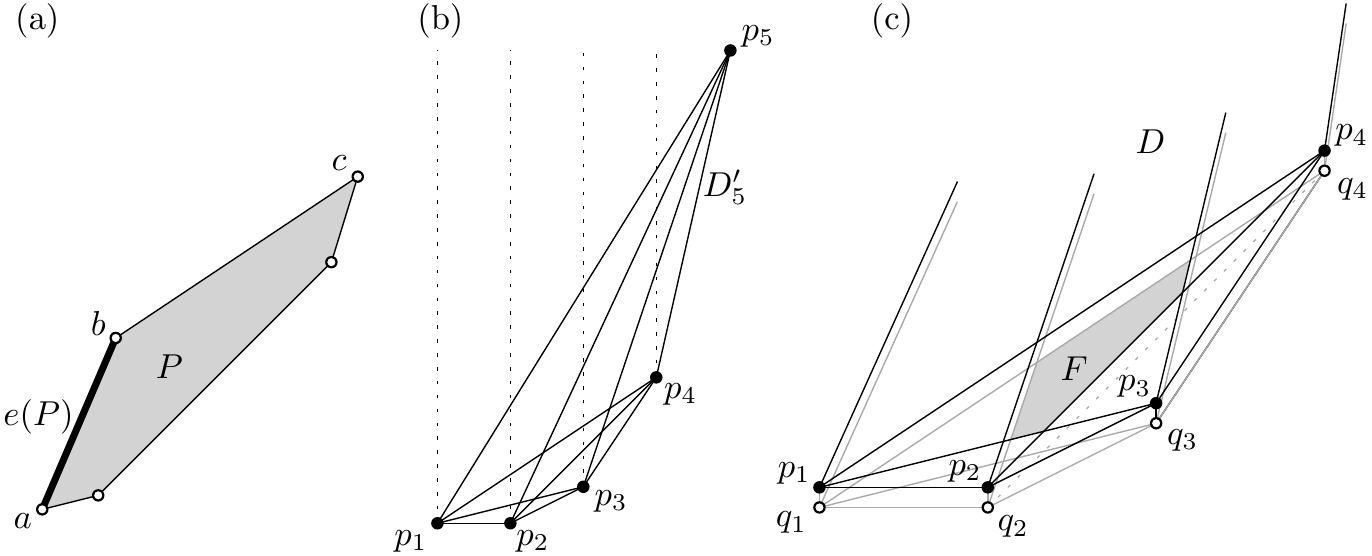}
    \caption{(a) A 4-cap free polygon $P$ that is not 3-cap free. (b) An example of the drawing $D'_m$ for $m=5$. If the point $p_m$ is chosen sufficiently high above $C_{m-1}$, then each line segment $\overline{p_i p_m}$ with $i < m$ is very close to the vertical line containing $p_i$ and thus all faces of $D'_m$ will be 4-cap free.  (c) An example of a part of the drawing $D$ with an example of a face $F$ of $D$. The line segments that are not in $D'_m$ are drawn grey and the non-edge from $X$ is dotted.}
    \label{fig-drawing}
\end{figure}

Let $m = n/2$.
First, we inductively construct a set $C_m$ of $m$ points $p_1,\dots,p_m$ in the plane that form a cup and their $x$-coordinates satisfy $x(p_i) = i$.
We let  $C_1=\{(1,0)\}$ and $C_2=\{(1,0),(2,0)\}$.
Now, assume that we have already constructed a set $C_{m-1} = \{p_1,\dots,p_{m-1}\}$ for some $m \geq 3$.
Let $D'_{m-1}$ be the drawing of the complete graph with vertices $p_1,\dots,p_{m-1}$.
We choose a sufficiently large number $y_m$, and we let $p_m$ be the point $(m,y_m)$ and $C_m = C_{m-1} \cup \{p_m\}$.
We also let $D'_m$ be the drawing of the complete graph with vertices $p_1,\dots,p_m$.
The number $y_m$ is chosen large enough so that the following three conditions are satisfied:
\begin{enumerate}
\item for every $i=1,\dots,m-1$, every intersection point of two line segments spanned by points of $C_{m-1}$ lies on the left side of the line $\overline{p_ip_m}$ if and only if it lies to the left of the vertical line $x=i$ containing the point $p_i$, 
\item if $F$ is a 4-cap free face of $D'_m$ that is not 3-cap free, then there is no point $p_i$ below the (relative) interior of $e(F)$,

\item no crossing of two edges of $D'_m$ lies on the vertical line containing some point $p_i$.
\end{enumerate}

Note that choosing the point $p_m$ is indeed possible as
choosing a sufficiently large $y$-coordinate $y_m$ of $p_m$ ensures that for each $i$, all the intersections of the line segments $p_ip_m$ with line segments of $D'_{m-1}$ lie very close to the vertical line $x=i$ containing the point $p_i$.
Note that it follows from the construction that no line segment of $D'_m$ is vertical and that no point is an interior point of more than two line segments of $D'_m$.
The drawing $D'_m$ also satisfies the following claim.

\begin{claim}
  \label{claim:drawing-faces}
Each inner face of $D'_m$ is a 4-cap free convex polygon.
\end{claim}
\begin{proof}
We prove the claim by induction on $m$.
The statement is trivial for $m \leq 3$ so assume $m \geq 4$.
Now, let $F$ be an inner face of $D'_{m-1}$.
By the induction hypothesis, $F$ is a 4-cap free convex polygon.
If $F$ is 3-cap free, then, by the choice of $p_m$, the line segments $\overline{p_i p_m}$ 
split $F$ into 4-cap free polygons (with the leftmost one being actually 3-cap free); see Figure~\ref{fig-drawing}(b).
If $F$ is 4-cap free and not $3$-cap free, then the choice of $p_m$ guarantees that the line segments $\overline{p_i p_m}$ 
split $F$ into 4-cap free polygons.
This is because the leftmost such polygon contains the whole edge $e(F)$ as there is no $p_i$ below the edge $e(F)$.
It remains to consider the inner faces of $D'_m$ which lie outside of the convex hull of the points $p_1,\dots,p_{m-1}$. These faces lie inside the triangle $p_1 p_{m-1} p_m$.
They are all triangular and therefore satisfy the claim; see the three faces with the topmost vertex $p_m=p_5$ in Figure~\ref{fig-drawing}(b).
\end{proof}

For a (small) $\varepsilon>0$ and every $i=1,\dots,m=n/2$, we let $q_i$ be the point $(i,y_i-\varepsilon)$.
That is, $q_i$ is a point slightly below $p_i$.
We choose $\varepsilon$ sufficiently small so that decreasing $\varepsilon$ to any smaller positive real number does not change the combinatorial structure of the intersections of the line segments spanned by the $n$ points $p_1,\dots,p_m,q_1,\dots,q_m$.
We let $D=D_n$ be the drawing with $n$ vertices $p_1,\dots,p_m,q_1,\dots,q_m$
containing all the line segments between two vertices, except the line segments $\overline{q_iq_j}$ where $i$ and $j$ are both even.

We have to show that at least quadratically many obstacles are needed to block all non-edges of $G$ in $D$.
Let $X$ be the set of line segments $\overline{q_i q_j}$ where both $i$ and $j$ are even.
Note that each line segment from $X$ corresponds to a non-edge of $G$ and that  
$|X|=\binom{\lfloor  n/4 \rfloor}{2} \geq n^2/40$ for a sufficiently large $n$.

\begin{claim}
  \label{claim:drawing-nonedges}
Every face of $D$ is intersected by at most two line segments from $X$.
\end{claim}
\begin{proof}
The outer face of $D$ is intersected by no line segment of $X$.
Every inner face of $D$ is contained in some inner face of $D'_m$ or in one of the parallelograms $q_iq_{i+1}p_{i+1}p_i$.

First, suppose that a face $F$ of $D$ is contained in some face $F'$ of $D'_m$.
Due to the choice of $\varepsilon$, a line segment $\overline{q_iq_j}$ of $X$ intersects $F'$ (if and) only if a part of the edge $\overline{p_ip_j}$ bounds $F'$ from above. By Claim~\ref{claim:drawing-faces}, at most two line segments bound $F'$ from above. We conclude that at most two line segments of $X$ intersect $F'$.
Consequently, also at most two line segments of $X$ intersect $F$.

Suppose now that a face $F$ of $D$ is contained in some parallelogram $Z=q_iq_{i+1}p_{i+1}p_i$. Let $j$ be the even integer which is equal to $i$ or to $i+1$.
Due to the construction, all line segments of $X$ intersecting $Z$ are incident to $q_j$.
Since edges of $G$ and line segments of $X$ incident to $q_j$ alternate around $q_j$, at most one line segment of $X$ intersects $F$.
\end{proof}

Now, let $\mathcal{O}$ be a set of obstacles such that $D$ and $\mathcal{O}$ form an obstacle representation of~$G$.
Then each non-edge of $G$ corresponds to a line segment between vertices of $D$ that is intersected by some obstacle from $\mathcal{O}$.
In particular, each line segment from $X$ has to be intersected by some obstacle from $\mathcal{O}$.
Each obstacle $O$ from $\mathcal{O}$ lies in some face $F_O$ of $D$ as it cannot intersect an edge of $G$.
Thus, $O$ can intersect only line segments from $X$ that intersect~$F_O$. 
It follows from Claim~\ref{claim:drawing-nonedges} that each obstacle from $\mathcal{O}$ intersects at most two line segments from $X$.
Since $|X| \geq n^2/40$, we obtain $|\mathcal{O}| \geq |X|/2 \geq n^2/80$.
Consequently, ${\rm obs}(D) \geq n^2/80$, which finishes the proof for $n$ even.
If $n$ is odd, we use the above construction for $n-1$ and add an isolated vertex to it.

\section{Obstacle Number is FPT Parameterized by Vertex Cover Number}
\label{sec-fpt}

In this section we show that computing the obstacle number is
fixed-parameter tractabile with respect to the vertex cover number.
As our first step towards a fixed-parameter algorithm, we show that,
for every graph, its obstacle number is upper-bounded by a function of
its vertex cover number.

\begin{lemma}
  \label{lem:vcbound}
  Let $G$ be a graph with vertex cover number~$k$.  Then $G$ admits an
  obstacle representation with $1+\binom{k}{2}+k\cdot 2^k$ obstacles.
\end{lemma}
\begin{proof}
  Our proof is by construction.  Let $X\subseteq V$
  be a cardinality-$k$ vertex cover of $G=(V,E)$.  We begin by placing
  the vertices in~$X$ on an arc $\gamma$ that is a small part of a
  circle with some center~$c$.  For an example, see Figure~\ref{fig:kernel},
  where $X=\{x,y,z\}$.  This ensures that $X$ is in general position.
  For each pair of non-adjacent vertices $x,y\in X$, we then place a
  single ``tiny'' point-like obstacle in the immediate vicinity of
  either $x$ or $y$; the choice is arbitrary.  Crucially, these at
  most $\binom{k}{2}$ many obstacles will not obstruct any visibility
  between other pairs of vertices in~$G$.  (In Figure~\ref{fig:kernel},
  these are the two square obstacles.)

  \begin{figure}[htb]
    \centering
    \includegraphics{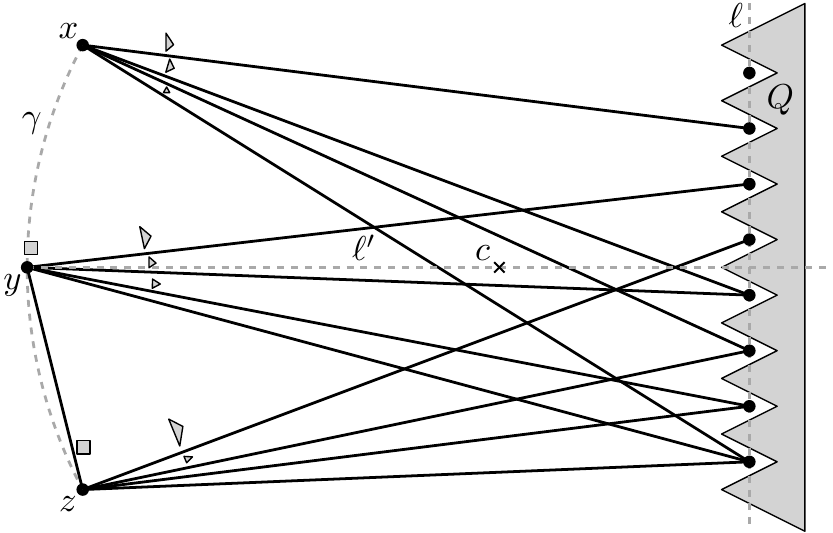}
    \caption{The obstacle number is bounded by a function of the
      vertex cover number.}
    \label{fig:kernel}
  \end{figure}
  
  Next, we partition the vertices in~$V\setminus X$ into at most $2^k$
  equivalence classes called \emph{types} based on the following
  equivalence: two vertices $a,b\in V\setminus X$ have the same type
  if and only if they have the same neighborhood.  We place these
  vertices on a line $\ell$ sufficiently far from $\gamma$.  On this
  line, we mark $2^k$ sufficiently small, pairwise disjoint
  segments---one for each type.  We place~$\ell$ orthogonally to the
  line~$\ell'$ that connects the center of $\gamma$ to $c$.
  The small segments will be placed in the general vicinity of the
  intersection point of~$\ell$ and $\ell'$.
  The idea is that the vertices of some type
  will all be placed on the segment of $\ell$ dedicated to that type.
  (In Figure~\ref{fig:kernel}, there is only one representative of each type.)

  We now construct a single obstacle $Q$ that blocks the pairwise
  visibilities between all vertices placed on~$\ell$; see
  Figure~\ref{fig:kernel}.  The boundary of $Q$ consists of a zig-zag
  line whose zigs and zags all intersect~$\ell$ and whose first and
  last line segments are a bit longer than the others.  We close~$Q$
  by connecting the endpoints of the zig-zag line by a line segment
  that is parallel to~$\ell$.  Intuitively, the shape of~$Q$ can be
  described as a ``hair comb''.  By construction, $Q$ guarantees that
  no two vertices placed on~$\ell$ will be visible to each other.  At
  the same time, $Q$ does not obstruct visibilities between the
  vertices in~$X$ and those in $V\setminus X$.

  To complete the construction, it suffices to block the visibilities
  between the vertices in the types and the vertices in $X$. To this
  end, for each type $T \subseteq V \setminus X$ and each vertex $x\in X$ such that no
  vertex in~$T$ is adjacent to~$x$, we add a single nearly point-like
  obstacle in the immediate vicinity of $x$ that (i)~blocks visibility
  between $x$ and the line segment of $\ell$ dedicated to~$T$, but
  (ii)~does not block visibility between any other pair of vertices in
  the construction.  (In Figure~\ref{fig:kernel}, these are the triangular
  obstacles.)  To see that this is possible, observe that, since
  $\gamma$ is a circular arc, small obstacles can be placed in the immediate
  vicinity of vertices on~$\gamma$ in the direction of~$c$ without
  lying on lines connecting other vertices on~$\gamma$.

  This completes the construction.  We have used at most
  $1+\binom{k}{2}+k\cdot 2^k$ obstacles.
\end{proof}

Our algorithm crucially uses the exponential-time decision procedure
for the existential theory of the reals (ETR) by
Renegar~\cite{Renegar92a,Renegar92b,Renegar92c}.  An \emph{existential
  first-order formula about the reals} is a formula of the form
$\exists x_1\dots\exists x_m$ $\Phi(x_1,\dots,x_m)$, where $\Phi$
consists of Boolean combinations of order and equality relations
between polynomials with rational coefficients over the variables
$x_1,\dots,x_m$.  Renegar's result on the ETR can be summarized as
follows.

\begin{theorem}[Renegar~\cite{Renegar92a,Renegar92b,Renegar92c}]
  \label{thm:renegar}
  Given any existential sentence $\Phi$ about the reals, one can
  decide whether $\Phi$ is true or false in time
  \[
    (L\log L\log\log L) \cdot (PD)^{O(N)},
  \]
  where~$N$ is the number of variables in~$\Phi$,~$P$ is the
  number of polynomials in~$\Phi$, $D$ is the maximum total
  degree over the polynomials in~$\Phi$, and~$L$ is the maximum length
  of the binary representation over the coefficients of the polynomials
  in~$\Phi$.
\end{theorem}

Note that Renegar's algorithm decides whether there exist real values
for the variables in~$\Phi$ that make~$\Phi$ true; the algorithm does
not explicitly compute these values.

\begin{lemma}
  \label{lem:etr}
  Let $G$ be a graph with vertex cover number~$k$, and let $h$ be a
  positive integer.  Suppose that there is a computable function~$f$
  such that~$G$ has at most $f(k)$ vertices, then we can decide in FPT
  time with respect to~$k$ whether $G$ admits an obstacle
  respresentation with at most $h$ obstacles.
\end{lemma}

\begin{proof}
  We show that an instance $(G,h)$ can be decided via ETR.  First,
  however, we establish a simple upper bound on the complexity of an
  obstacle representation for an arbitrary graph~$H$ with $n$ vertices
  and $m$ edges; that is, a bound on the total number of obstacle
  corner points in the representation.

  For now, we assume that the graph~$H$ is connected.  Observe that
  each obstacle corresponds to a face in a straight-line
  drawing~$\Gamma$ of~$H$.  From a planarization $\Gamma'$ of~$\Gamma$
  we can obtain an obstacle representation by using two copies of each
  edge of~$\Gamma'$ and slightly offsetting (and slightly modifying
  the length of) each copy into one of the adjacent faces.  If one of
  the endpoints of the edge is a leaf, the two copies go to the same
  face and are slightly rotated in order to meet each other.  As a
  result, each face of~$\Gamma'$ is turned into an obstacle polygon,
  and the number of corner points of the polygon is the same as the
  number of vertices on the boundary of the face.  Clearly, $\Gamma'$
  consists of at most $O(m^2)$ vertices (where at least two edges
  of~$\Gamma$ intersect) and $O(m^2)$ edges (the line segments that
  connect two vertices that are consecutive on one of the edges
  of~$\Gamma$).

  In case the graph~$H$ is not connected, we need to ``stitch
  together'' the obstacle representations of the connected components,
  which can be done with a constant number of extra obstacle corner
  points for each of the at most $n$ connected components of~$H$.
  (A component whose obstacle representation has no obstacle incident
  to the outer face can be shrunken and placed inside a nearly-closed
  cavity of an obstacle of another component.)

  Hence, if $N$ is the complexity of an obstacle representation for a
  graph with $n$ vertices and $m$ edges, then $N \in O(m^2+n)$.  Now
  it is clear that a collection~$Z$ of $h$ obstacles can be encoded by
  $h$ sequences of corner points:
  \begin{align*}
  Z = \Big\{ &\big\langle(x_{1,1},y_{1,1}),\dots,(x_{1,n_1},y_{1,n_1})\big\rangle,
  \;\big\langle(x_{2,1},y_{2,1}),\dots,(x_{2,n_2},y_{2,n_2})\big\rangle,\dots\\
  &\big\langle (x_{h,1},y_{h,1}),\dots,(x_{h,n_h},y_{h,n_h})\big\rangle \Big\},
  \end{align*}
  where $N=\sum_{i=1}^h n_i \in O(m^2+n)$ is the total complexity
  of~$Z$.
  
  We start to set up an ETR formula by encoding the coordinates of
  the $n$ vertices of the given graph~$G$ and the $N$ corner points of
  the $h$ obstacles in the plane:
  \[\exists (x_1,y_1), \dots, (x_n,y_n), (x_{1,1},y_{1,1}), \dots,
    (x_{h,n_h},y_{h,n_h}) \in \mathbb{R}^{2n+2N}\]
  This encoding represents a geometric drawing of the graph (where
  edges are line segments) and the polygons that describe the
  obstacles.  We insist that (i)~no edge of the graph intersects an
  edge formed by a pair of consecutive corner points of an obstacle,
  and (ii)~every non-edge of the graph does intersect an edge formed
  by a pair of consecutive corner points of an obstacle.  Both
  properties can easily be expressed using ETR.  The lemma statement
  assumes that~$G$ has at most $f(k)$ vertices.  This implies
  that $N$, the complexity of~$Z$, is bounded by $O(f^4(k))$.
  Therefore, Theorem~\ref{thm:renegar} yields an algorithm to decide
  an instance $(G,h)$ whose running time is FPT with respect to~$k$.
\end{proof}

Our proof of Theorem~\ref{thm:fpt} relies also on a Ramsey-type
argument based on the following formulation of Ramsey's
theorem~\cite{ramsey29}:

\begin{theorem}[\cite{ramsey29}]
\label{thm:ramsey}
There exists a function $\texttt{Ram}\colon\N^2\rightarrow \N$ with the following property. For each pair of integers $q,s$ and each clique $H$ of size at least $\texttt{Ram}(q,s)$ such that each edge has a single label (color) out of a set of $q$ possible labels, it holds that $H$ must contain a subclique~$K$ of size at least~$s$ such that every edge in $K$ has the same label.
\end{theorem}

We are now ready to prove Theorem~\ref{thm:fpt}, which we restate
here.

\thmfpt*

\newcommand{\typesize}[1]{\texttt{typesize}(#1)}
\newcommand{\cliquesize}[1]{\texttt{cliquesize}(#1)}
\newcommand{\crevicenum}[1]{\texttt{crevicenum}(#1)}

\begin{proof}
  Let $n$ be the number of vertices of~$G$, and let $k$ be the vertex
  cover number of~$G$.

  \proofparagraph{Setup.}
  First, we use one of the well-known fixed-parameter algorithms for
  computing a cardinality-$k$ vertex cover $X$ in $G$.

Second, we immediately output ``\texttt{Yes}'' if
$h \ge 1+\binom{k}{2}+k\cdot 2^k$.  This is correct due to
Lemma~\ref{lem:vcbound}.  For the rest of the proof, we assume that
$h \le \binom{k}{2}+k\cdot 2^k$.

Third, we recall the notion of ``type'' defined in the proof of
Lemma~\ref{lem:vcbound} and note that if $G$ contains no type of
cardinality greater than a function of $k$, then the size of $G$ is
upper-bounded by a function of the parameter. In this case, $G$ can be
solved via Lemma~\ref{lem:etr}.
Hence, we may proceed with the assumption that $G$ contains at least
one type~$T$ that is larger than a given function of~$k$.  We call
this function $\typesize$, and we need that
$\typesize{k}=\texttt{Ram}(h,15k\cdot 2^k)$.  (For the definition of
\texttt{Ram}, see Theorem~\ref{thm:ramsey}.)  Our proof strategy will
be to show that under these conditions, $G$ is a \texttt{Yes}-instance
(i.e., $G$ admits a representation with at most $h$ obstacles) if and
only if the graph $G'$ obtained by deleting an arbitrary vertex from
$T$ is a \texttt{Yes}-instance.  In other words, we may prune a vertex
from~$T$.  Once proved, this claim will immediately yield
fixed-parameter tractability of the problem: one could iterate the
procedure of computing a vertex cover for the input graph, check
whether the types are sufficiently large, and, based on this, one
either brute-forces the problem or restarts on a graph that contains
one vertex less. (Note that the number of restarts is at most~$n$.)

\proofparagraph{Establishing the Pruning Step.}
Hence, to complete the proof it suffices to establish the correctness
of the pruning step outlined above, i.e., that the instance $(G,h)$ is
equivalent to the instance $(G',h)$, where $G'$ is obtained from~$G$ by deleting a
single vertex of the type~$T$ identified above. Unfortunately, the
proof of this claim is far from straightforward. 

As a first step, consider a hypothetical ``optimal'' solution $S$ for $(G,h)$, i.e., $S$ is an obstacle representation of $G$ with the minimum number of obstacles. It is easy to observe that using the same obstacles for the graph~$G'$ obtained by \emph{removing} a single vertex from type~$T$ yields a desired solution $S'$ for $(G',h)$.
(Note that $S'$ might not be obstacle-minimal, but this is not an issue.) 
The converse direction is the difficult one: for that, it suffices to establish the equivalent claim that the graph~$G''$ obtained by \emph{adding} a vertex to type~$T$ admits an obstacle representation~$S''$ with the same number of obstacles as~$S$.

Given $S$, we now consider an auxiliary edge-labeled clique $H$
as follows. Assume that the obstacles in $S$ are numbered from $1$ to
$h_S$, where $h_S\leq h$, and that there is an arbitrary linear
ordering $\prec$ of the vertices of~$G$ (which will only be used for
symmetry-breaking). The vertices of~$H$ are precisely the vertices 
of type~$T$. The labeling of the edges of~$H$ is defined as follows. For each
pair of vertices $a,b\in T$ with $a\prec b$, the edge~$ab$ in~$H$ is
labeled by the integer $q\in [h_S]$, where $q$ is the first obstacle
encountered when traversing the line segment~$\overline{ab}$ from~$a$
to~$b$.

Crucially, by Theorem~\ref{thm:ramsey} and by the fact that $|H|\geq
\typesize{k}=\texttt{Ram}(h,15k\cdot 2^k)$, the clique~$H$ must
contain a subclique~$K$ of size at least $\cliquesize{k}=15k\cdot
2^k$ such that each edge in $K$ has the same label, say~$p$.
This means that, in particular, every edge in~$K$ intersects the
obstacle~$p$.

\proofparagraph{Analysis of an Obstacle Clique.}
Intuitively, our aim will be to show that the obstacle~$p$ can be
safely extended towards some vertex $z$ in $K$, where by ``safely'' we
mean that it neither intersects another obstacle nor blocks the
visibility of an edge; this extension will either happen directly from
$S$, or from a slightly altered version of $S$ as will be described
later. Once we create such an extension, it will be rather
straightforward to show that $p$ can be shaped into a tiny
``comb-like'' slot for a new vertex $z'$ next to $z$ which will have
the same visibilities (and hence neighborhood) as $z$, which means
that we have constructed a solution $S'$ for $G'$ as desired.

To this end, we now choose an arbitrary vertex $x\in X$ adjacent to
all vertices in~$K$.  Let a vertex~$v$ in~$K$ be \emph{$x$-separating}
if the ray~$\overrightarrow{xv}$ does not intersect $p$. As our first
step in this case, we will show that there are only very few
$x$-separating vertices in $K$.

\begin{claim}
\label{claim:fewsep}
There exist at most two $x$-separating vertices in $K$.
\end{claim}

\begin{proof}[Proof of the claim.]
Assume for a contradiction that there are three $x$-separating vertices in~$K$, say~$s$, $t$, and~$u$. 
The rays from~$x$ through these three vertices split the plane into three parts.
At most one of these parts can have an angle of over $180^\circ$ at~$x$. 
The other two parts must both contain $p$ since~$p$ must occur on the line segments between each pair of~$s$, $t$, and~$u$. 
However, by definition, $p$ cannot cross the rays defining these parts. 
This contradicts the fact that $p$ is a single simply-connected obstacle.
\end{proof}

\proofparagraph{Crevices.} %
Now, we consider the (e.g., clockwise) cyclic order of vertices in~$K$
around~$x$.  By Claim~\ref{claim:fewsep}, there must exist a sequence
of at least $(\cliquesize{k}-2)/{2}$ vertices in $K$ which occur
consecutively in this order and are all non-separating; let these be
denoted by~$v_1,\dots,v_g$ for some
$g\ge (\cliquesize{k}-2)/2$.  Recall that for every $i\in[g]$,
the ray $\overrightarrow{x v_i}$ intersects no obstacle before
reaching~$v_i$, and after passing~$v_i$, it intersects the
obstacle~$p$.  Let~$p_i$ be the intersection point of
$\overrightarrow{x v_i}$ with $\partial p$, the boundary of~$p$.
Let~$p_{i,i+1}$ be the first intersection point of
$\overrightarrow{v_iv_{i+1}}$ with~$\partial p$, and let~$p_{i+1,i}$
be the first intersection point of $\overrightarrow{v_{i+1}v_i}$
with~$\partial p$.
If~$v_i \prec v_{i+1}$, then $p$ is the first obstacle on the ray
$\overrightarrow{v_iv_{i+1}}$.  Otherwise $p$ is the first obstacle on
the ray $\overrightarrow{v_{i+1}v_i}$.
Note that a counterclockwise traversal of~$\partial p$ starting
in~$p_1$ visits
$p_{1,2}, p_{2,1}, p_2, p_{2,3}, \dots, p_{g,g-1}, p_g$ (in this
order) because $p$ is a simply-connected region.  For~$1 < j < g$ and
$j'\in\{j-1,j+1\}$, let $p'_{j,j'}$ be the first intersection point of
the ray $\overrightarrow{xp_{j,j'}}$ and~$p$, where
$p'_{j,j'}=p_{j,j'}$ is possible.  (For example, in
Figure~\ref{fig:crevices2}, $p_{2,3}' \ne p_{2,3}$, but
$p_{3,2}' = p_{3,2}$.)  Now we can define the following polygons for
$1 < i < g$; see Figure~\ref{fig:crevices2}:
\begin{figure}
  \centering
  \includegraphics{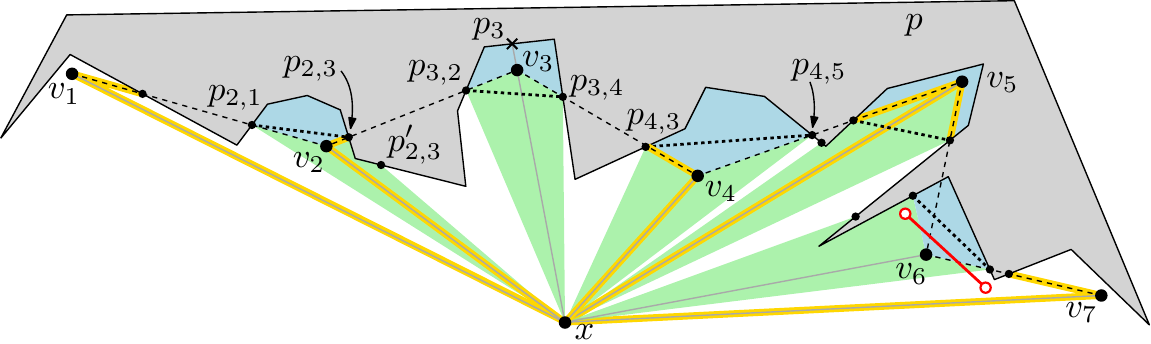}
  \caption{The crevices of obstacle~$p$ (with respect to vertex~$x$
    in~$X$) are the light blue regions that contain vertices $v_2$,
    $v_3$, and $v_4$ of~$K$.  Each extended crevice consists of a
    light blue and an adjacent light green region.  The dotted line
    segments are the doors of the corresponding vertices.  The red
    line segment pierces the crevice~$Z_6$.  The orange L-shapes
    assume
    $v_1 \prec v_2 \prec v_3 \succ v_4 \succ v_5 \prec v_6 \succ v_7$.
    With $\partial p$, two consecutive L-shapes enclose parts of at
    most three extended crevices.}
  \label{fig:crevices2}
\end{figure}
\begin{itemize}
\item The \emph{crevice} of~$v_i$ is a polygon~$Z_i$ whose boundary
  starts with the line segment $\overline{v_i p_{i,i+1}}$, then
  follows $\partial p$ in the clockwise direction, passes the point~$p_i$,
  reaches the point~$p_{i,i-1}$, and returns to~$v_i$ along the line
  segment $\overline{p_{i,i-1} v_i}$.
\item The \emph{extended crevice} of~$v_i$ is a polygon~$Z_i'$ whose
  boundary starts with the line segment $\overline{x p'_{i,i+1}}$,
  then follows~$\partial p$ in the clockwise direction, passes~$p_i$,
  reaches~$p'_{i,i-1}$, and returns to~$x$ along the line segment
  $\overline{p'_{i,i-1} x}$.
\end{itemize}
We consider (extended) crevices as open; hence they do not intersect
the polygon~$p$ to which they belong.  We use the term ``crevice''
here since, for every $i \in [g]$, the vertex~$v_i$ sits in an
indentation of~$p$; see Figure~\ref{fig:crevices2}.  Note that there
are $g-2$ extended crevices and that each of them contains precisely
one crevice.  No two extended crevices intersect, and the (extended)
crevices are naturally ordered according to the numbering of the
vertices from the perspective of~$x$ as described above.

The reason why we define these (extended) crevices is that they will
allow us to slightly alter~$S$ to accommodate the before-mentioned
additional vertex of~$T$.  To do so, we need to carefully avoid other
obstacles as well as visibility edges.  In the remainder of this
part of the proof, we identify three types of ``bad'' (extended)
crevices and show that each of these occurs only a bounded number of
times.

First, we simply observe that at most $k-1$ extended crevices can
contain a vertex of the vertex cover~$X$ (because
$|X \setminus \{x\}|=k-1$).

Second, we argue that every obstacle intersects only very few extended
crevices (of~$p$ with respect to~$x$).  To this end, observe that for
every $i \in [g]$, the line segment $\overline{xv_i}$ is an edge
of~$G$ and as such it must not intersect any obstacle.  Moreover, for every
$i\in[g-1]$ with $v_i \prec v_{i+1}$, by the definition of~$K$, no
obstacle intersects the (open) line segment $\overline{v_ip_{i,i+1}}$.
This line segment, together with $\overline{xv_i}$, forms an L-shape
that does not intersect any obstacle (see, e.g., the orange L-shape
at~$v_4$ in Figure~\ref{fig:crevices2}).  Similarly, for every
$1<i\le g$ with $v_i \prec v_{i-1}$, the line
segment~$\overline{xv_i}$ and the open line segment
$\overline{v_ip_{i,i-1}}$ form an L-shape that does not intersect any
obstacle (see, e.g., the orange L-shape at~$v_2$ in
Figure~\ref{fig:crevices2}).  Together with~$\partial p$, two
consecutive L-shapes form a polygon whose boundary does not intersect
any obstacles (other than~$p$).  The interior of such a polygon
contains parts of at most one extended crevice (if
$v_{i-1} \succ v_i \prec v_{i+1}$), parts of two extended crevices (if
$v_{i-1} \prec v_i \prec v_{i+1}$ or
$v_{i-1} \succ v_i \succ v_{i+1}$), or parts of three extended
crevices (if $v_{i-1} \prec v_i \succ v_{i+1}$); see
Figure~\ref{fig:crevices2}.  If $v_1 \succ v_2$ and/or $v_{g-1} \prec
v_g$, the parts of the extended crevices of~$v_2$ and/or~$v_{g-1}$ lie
outside the region enclosed by~$p$, the first L-shape, and the last
L-shape.  Hence any obstacle outside that region can intersect at most
two extended crevices.
Summarizing the above, we conclude that every obstacle can intersect
at most three extended crevices.  This implies that at most
$3h \le 3 \cdot \big(\binom{k}{2}+k\cdot 2^k\big)$ extended crevices
are intersected by obstacles.

Third, for $i\in[g]$, let the \emph{defining triangle} of
crevice~$Z_i$ be the triangle $\triangle(v_i,p_{i,i+1},p_{i,i-1})$,
where $v_i$ is the \emph{tip} of that triangle and the line segment
$\overline{p_{i,i+1} p_{i,i-1}}$ is its \emph{door}.  (In
Figure~\ref{fig:crevices2}), the doors are dotted.)  Note that
the line segment $\overline{xv_i}$ may or may not cross the door
of~$v_i$.
If $\overline{xv_i}$ does not cross the door (as in the case of~$v_2$
and of~$v_4$ in Figure~\ref{fig:crevices2}), there could exist an
edge~$e$ of~$G$ that crosses both non-door line segments of the
defining triangle.  In this case, we say that the crevice~$Z_i$ is
\emph{pierced} (by the edge~$e$).  For example, the red solid line
segment in Figure~\ref{fig:crevices2} pierces~$Z_4$.

\begin{claim}
  \label{claim:pierced}
  Every sequence of at least $2k-1$ crevices contains at least one
  crevice that is not pierced.
\end{claim}

\begin{proof}[Proof of the claim.]
  Assume, for a contraction, that there exists an index~$j$ with
  $1<j<g-2k+2$ such that each of the $2k-1$ crevices in the sequence
  $Z_j, Z_{j+1}, \dots, Z_{j+2k-2}$ is pierced.  This implies that,
  for $j \le i \le j+2k-2$, when going from $v_{i-1}$ to~$v_{i+1}$,
  there is a left turn in~$v_i$; see Figure~\ref{fig:wedges}.  Hence
  the set $\{v_{j-1},v_j, \dots, v_{j+2k-2}, v_{j+2k-1}\}$ is the
  vertex set of a convex polygon (dashed in Figure~\ref{fig:wedges})
  that contains the doors of $v_j, \dots, v_{j+2k-2}$.
  Clearly, $x$ lies outside this convex polygon.
  
  Consider an arbitrary vertex~$u$ in the vertex cover~$X$ of~$G$.  We
  show that at most two edges incident to~$u$ pierce a crevice in our
  sequence.  There are two ways how an edge~$uv$ can pierce a
  crevice~$Z_i$ in that sequence.  When rotating~$uv$ counterclockwise
  around~$u$, $uv$ can either first hit the tip~$v_i$ or first hit the
  door of~$v_i$.  In the former case, we say that $uv$
  \emph{tip-pierces}~$Z_i$; in the latter case, $uv$
  \emph{door-pierces}~$Z_i$.  We show that $u$ can have at most one
  edge of either type.  The arguments are symmetric, so it suffices to
  show that $u$ is incident to at most one edge that tip-pierces a
  crevice in our sequence.  To this end, we define, for each index~$i$
  with $j \le i \le j+2k-2$, an open wedge~$W_i$ with apex~$v_i$ that
  is bounded by the ray~$\overrightarrow{v_iv_{i+1}}$ and by the ray
  that starts in~$v_i$ and is
  opposite to the ray $\overrightarrow{v_iv_{i-1}}$; see
  Figure~\ref{fig:wedges}.  Note that, for different indices~$i$
  and~$i'$, wedges~$W_i$ and~$W_{i'}$ are disjoint.

  \begin{figure}[htb]
    \centering
    \includegraphics{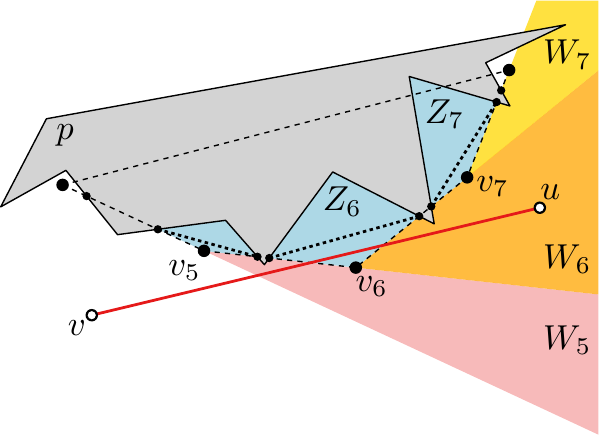}
    \caption{Wedges $W_5$, $W_6$, $W_7$ for crevices $Z_5$, $Z_6$, and
      $Z_7$.  In order for an edge~$uv$ to tip-pierce~$Z_6$, $u$ must
      lie in wedge~$W_6$.}
    \label{fig:wedges}
  \end{figure}
  
  We claim that, for an edge $uv$ to tip-pierce the crevice~$Z_i$,
  vertex~$u$ has to lie in wedge~$W_i$.  This is due to the fact that
  the ray $\overline{uv}$ must enter the defining triangle of~$v_i$
  via the edge~$\overline{v_iv_{i,i+1}}$ and leave the defining
  triangle via the edge $\overline{v_iv_{i,i-1}}$.  If~$u$ does not
  lie in~$W_i$ but in some~$W_{i'}$ with $i'>i$, then $u$ does not
  see~$v_i$ nor the edge~$\overline{v_iv_{i,i+1}}$. If~$u$ lies in
  some~$W_{i'}$ with $i'<i$, then $u$ does not see the
  edge~$\overline{v_iv_{i,i-1}}$ through the
  edge~$\overline{v_iv_{i,i+1}}$.

  Given our claim and the fact that the wedges are pairwise disjoint
  yields that $u$ is incident to at most one edge that tip-pierces a
  crevice (and, similarly, at most one edge that door-pierces a
  cervice.  Hence, any vertex in~$X \setminus \{x\}$ is incident to at
  most two edges that pierce crevices.  Since~$x$ sees
  $v_1,\dots,v_g$, no edge incident to~$x$ pierces a crevice.  Hence,
  at most $2k-2$ of the crevices in the sequence
  $Z_j, Z_{j+1}, \dots, Z_{j+2k-2}$ can be pierced, which yields the
  desired contradiction.

\end{proof}

Claim~\ref{claim:pierced} implies that every subsequence
$(v_{j},v_{j+1},\dots,v_{j+2k})$ of consecutive vertices in~$K$ must
contain at least one crevice that is not pierced.  Let an extended
crevice $Z_i'$ be \emph{good} if~$Z_i'$ is not intersected by any
obstacle, $Z_i' \cap X = \emptyset$, and $Z_i$ is not pierced.  
Recall that $\cliquesize{k}=15 k \cdot 2^k$.  In
view of the bounds established above, we conclude that there
must exist a sequence of good extended crevices of length at least
\begin{align*}
  \crevicenum{k} &= \left(
                   \left( \frac{\cliquesize{k}-2}{2}-2 \right) -
                   \left( (k-1) + 3\cdot\binom{k}{2}
                   + 3 \cdot k 2^k \right)
                   \right) \cdot \frac{1}{2k} \\
                 &\ge \frac{ (7.5 \cdot k 2^k - 3) - (1.5 \cdot k 2^k + 3\cdot k2^k) }{2k} 
                   \qquad \text{using $k \le 2^k$ } \\
                 &\ge \frac{ 6.5 \cdot k 2^k - 4.5 \cdot k 2^k }{2k}
                   \ge 2^k \qquad \text{using $k2^k \ge 3$ for $k \ge 2$}
\end{align*}

\proofparagraph{Extending $p$ in a Non-Obstructive Way.} %
From now on, we focus on good extended crevices in the
aforementioned, sufficiently long sequence
$Z'_{\sigma(1)},Z'_{\sigma(2)},\dots,Z'_{\sigma(\crevicenum{k})}$,
where the $\sigma(i)$-th extended crevice $Z_{\sigma(i)}$ is defined
by the vertex $v_{\sigma(i)}$ and the points
$p_{\sigma(i),\sigma(i)+1}$, $p_{\sigma(i),\sigma(i)-1}$.  Our aim in
this section is to show that there exists some solution $S^\circ$
(which will either be $S$, or a slight adaptation of $S$) for $(G,h)$
which contains a crevice $Z_{\sigma(i)}$ where we can extend~$p$ so as
to accommodate an additional vertex of type~$T$ without intersecting
visibility lines between any pair of vertices that are adjacent
in~$G$.

We say that a good crevice $Z'_{\sigma(i)}$ is \emph{perfect} if it
furthermore contains no vertex other than $v_{\sigma(i)}$.  Next we
claim that if we can show that $(G,h)$ admits a solution with a
perfect crevice, we are essentially done.  Recall that the graph~$G''$
has vertex set $V(G'') = V(G) \cup \{v\}$, where $v$ is the additional
vertex of type~$T$.

\begin{claim}
  \label{claim:perfectdone}
  If $(G,h)$ admits a solution 
  with a perfect crevice, then $(G'',h)$ also admits a solution.
\end{claim}

\begin{proof}[Proof of the claim.]
  Let $i \in [\crevicenum{k}]$ be such that $Z'_{\sigma(i)}$ is a
  perfect crevice in a solution of~$(G,h)$.  We modify the
  obstacle~$p$ that is adjacent to~$Z'_{\sigma(i)}$, in order to
  accommodate the additional vertex~$v$ of type~$T$.
  We add to~$p$ a ``thick'' line segment that is orthogonal to the
  door of~$v_{\sigma(i)}$ and stops very close to~$v_{\sigma(i)}$,
  slightly to its right, say; see Figure~\ref{fig:perfect0}.
  \begin{figure}[htb]
    \centering
    \includegraphics[page=3]{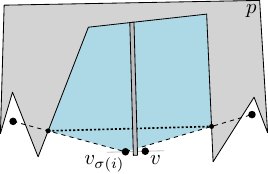}
    \caption{Adding a vertex of type~$T$ to a perfect crevice $Z'_{\sigma(i)}$}
    \label{fig:perfect0}
  \end{figure}
  Then we place~$v$ to the right of this extension of~$p$ such
  that~$v_{\sigma(i)}$ and~$v$ do not see each other.  By making the
  extension thin and short enough and by placing~$v$ sufficiently
  close to~$v_{\sigma(i)}$, we can ensure that $v$ sees the same
  vertices as~$v_{\sigma(i)}$ (namely those adjacent
  to~$v_{\sigma(i)}$ in~$G$).  Since the extended crevice contains no
  other vertices and does not intersect any obstacles, this extension
  could only block a line of sight between two vertices outside of
  $Z'_{\sigma(i)}$, but this line would pierce~$Z_{\sigma(i)}$,
  contradicting the fact that~$Z'_{\sigma(i)}$ is good.  Hence, we
  have obtained a solution for $(G'',h)$.
\end{proof}

As a consequence, if a solution already contains a perfect crevice, we
can invoke the previous claim to complete the proof.  Unfortunately,
it may happen that none of the at least $2^k$ good crevices in a
solution is perfect.  We now deal with this case.

\begin{claim}
  If the instance $(G,h)$ admits a solution, then $(G,h)$ admits also
  a solution with a perfect crevice.
\end{claim}

\begin{proof}[Proof of the claim.]
  Let $S$ be a solution of $(G,h)$.
  If $S$ does not contain a perfect crevice, then each of the at least
  $2^k$ good crevices in~$S$ must contain some vertex that is not part
  of the vertex cover~$X$.  Let~$\mathcal{Z}$ be the set of
  (non-extended) crevices that correspond to the good crevices in~$S$.
  We start by arguing that each crevice~$Z_{\sigma(i)}$
  in~$\mathcal{Z}$ contains only a single vertex of type~$T$,
  specifically~$v_{\sigma(i)}$.  Indeed, assume for a contradiction
  that~$Z_{\sigma(i)}$ contains not only~$v_{\sigma(i)}$ but also
  another vertex~$w$ of the same type as~$v_{\sigma(i)}$.
  Since~$Z_{\sigma(i)}$ corresponds to a good crevice, $Z_{\sigma(i)}$
  does not intersect any obstacle, and hence the visibility
  between~$v_{\sigma(i)}$ and~$w$ must be blocked by~$p$.  Similarly,
  the first obstacle hit when traversing the line segment from~$w$ to
  the vertex~$v_{\sigma(i)+1}$ in the next crevice~$Z_{\sigma(i)+1}$
  in our sequence must again be~$p$ since this occurs on the boundary
  of~$Z_{\sigma(i)}$.  As a consequence, we obtain that $w$ is itself
  contained in a good crevice for $S$, contradicting the choice of
  $\mathcal{Z}$.

  Next, we mark all of the crevices in $\mathcal{Z}$ as unprocessed
  and then iterate the following procedure.  We choose an unprocessed
  crevice $Z_{\sigma(i)}\in \mathcal{Z}$ and consider the vertex~$u$
  in~$Z_{\sigma(i)}$ that is farthest from the door in the direction
  away from $x$ (i.e., ``inside'' the crevice). Let the
  \emph{$u$-door} be a line segment defined as follows.  The $u$-door
  is parallel to the door of~$v_{\sigma(i)}$, contains~$u$, and starts
  and ends at the first intersection with the boundary of~$p$ (i.e.,
  the $u$-door extends from~$p$ in both directions until it hits the
  boundary of the crevice; see the dotted line segment through~$u$ in
  Figure~\ref{fig:perfect1}).  By the choice of $u$, the part of
  $Z_{\sigma(i)}$ that is separated from $x$ by the $u$-door (darker
  blue in Figure~\ref{fig:perfect1}) contains no vertex at all.

  \begin{figure}
    \begin{subfigure}{.33\textwidth}
      \centering
      \includegraphics[page=1]{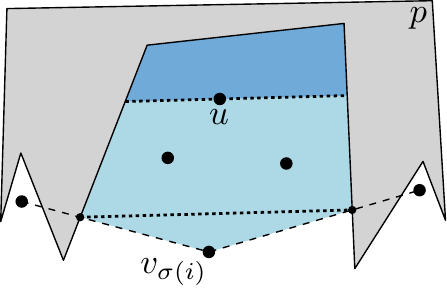}
      \caption{$u$ lies in the crevice of $v_{\sigma(i)}$}
      \label{fig:perfect1}
    \end{subfigure}
    \hfill
    \begin{subfigure}{.58\textwidth}
      \centering
      \includegraphics[page=2]{perfect-crevice}
      \caption{accommodating all vertices of the same type as~$u$}
      \label{fig:perfect2}
    \end{subfigure}
    \caption{Reshaping the obstacle $p$ to establish the existence of a
      perfect extended crevice.}
    \label{fig:perfect}
  \end{figure}

  Hence, it is possible to extend the boundary of $p$ towards $u$ in
  this part of the crevice without intersecting any visibility lines
  of edges in~$G$.  Let~$T_u$ be the set of all vertices of the same
  type as~$u$ that occur in the crevices in $\mathcal{Z}$.  We now
  extend~$p$ by constructing a tiny haircomb-like structure in the
  immediate vicinity of $u$ with precisely $|T_u|$ slots; then we move
  each vertex in~$T_u$ to a dedicated slot; see
  Figure~\ref{fig:perfect2}.
  The teeth of the haircomb are long enough to block the
  visibilities between every pair of vertices in~$T_u$,
  and they are short enough to preserve the
  visibilities between all other vertices in the graph (since the part
  of $Z_{\sigma(i)}$ behind the $u$-door is empty, and since all
  vertices are placed in the immediate vicinity of~$u$).

  At this point, we mark the crevice $Z_{\sigma(i)}$ and also the
  type~$T_u$ as ``processed'', and restart our considerations with the
  newly constructed solution and some unprocessed crevice.  After $j$
  steps of this procedure, we have processed precisely $j$ crevices,
  and none of the unprocessed crevices contains any vertex of a
  processed type.  Since the number of types other than~$T$ is
  upper-bounded by $2^k-1$, the procedure is guaranteed to construct a
  solution where at least one unprocessed
  crevice~$Z_{\sigma(i^\star)}$ contains only a single
  vertex~$v_{\sigma(i^\star)}$ of type~$T$.  By definition,
  $Z_{\sigma(i^\star)}$ corresponds to a perfect
  crevice~$Z'_{\sigma(i^\star)}$.
\end{proof}

In summary, we conclude that if $(G,h)$ is a \texttt{YES}-instance,
then it must admit a solution~$S$ with a perfect crevice.  By
Claim~\ref{claim:perfectdone}, this implies that $(G'',h)$, where
$G''$ is obtained by adding a vertex of type~$T$ to~$G$, is also a
\texttt{YES}-instance.  Moreover, we have already argued that if
$(G'',h)$ is a \texttt{YES}-instance, then so is $(G,h)$.  Hence, if
$G$ contains at least $\typesize{k}$ vertices of type~$T$, we can
delete a vertex of type~$T$ to obtain an equivalent instance.
Iterating this step results in an equivalent instance $(H,h)$ where
there are at most $\typesize{k}$ vertices of each type.  In
particular, $H$ contains at most $k+(2^k\cdot \typesize{k})$ many
vertices.  At this point, we can solve the instance by brute force in
the manner described in Lemma~\ref{lem:etr}.
\end{proof}

\section{NP-Hardness of Deciding Whether a Given Obstacle is Enough}
\label{sec-hardness}

In this section, we investigate the complexity of the obstacle
representation problem for the case that a single (outside) obstacle
is given.  In particular, we prove Theorem~\ref{thm-hardness}, which
we restate here.

\hardness*

\begin{proof}
  We reduce from \textsc{3-Partition}, which is NP-hard
  \cite{gj-crmsr-75}.  An instance of \textsc{3-Partition} is a
  multiset~$S$ of $3m$ positive integers $a_1,\dots,a_{3m}$, and the
  question is whether~$S$ can be partitioned into $m$ groups of
  cardinality~3 such that the numbers in each group sum up to the same
  value $B$, which is $(\sum_{i=1}^{3m} a_i)/m$.
  The problem remains NP-hard if $B$ is polynomial in~$m$ and 
  if, for each $i \in [3m]$, it holds that $a_i \in (B/4,B/2)$.
  
  Given a multiset~$S$, we construct a graph~$G$ and a simple 
  polygon~$P$ with the property that the vertices of~$G$ can be placed in~$P$
  such that their visibility graph with respect to~$P$ is isomorphic 
  to~$G$ if and only if $S$ is a yes-instance of \textsc{3-Partition}.

  Let $G$ be a clique with vertices $v_1,\dots,v_{3m}$ 
  where, for $i \in [3m]$, \emph{clique vertex}~$v_i$ is connected to $a_i$ 
  \emph{leaves} $\ell_{i,1},\dots,\ell_{i,a_i}$ (that is, vertices of
  degree~1).  Note that $G$ has $3m+Bm$ vertices and $\binom{3m}{2}+Bm$
  edges.  The polygon~$P$ (see Figure~\ref{fig:hardness}) is an
  orthogonal polygon with $m$ groups of $B$ bays (red in
  Figure~\ref{fig:hardness}) that are separated from each other by
  corridors (of height~$B/4$ and width~$4B$; light blue in
  Figures~\ref{fig:hardness} and~\ref{fig:polygon-size}).
  Each bay is a unit square; any two consecutive bays are one unit
  apart.  The height of~$P$ (including the bays) is $B/2+1$.
  Hence, the sizes of~$G$ and~$P$ are polynomial in~$m$.

  \begin{figure}[htb]
      \centering
      \includegraphics[scale=0.99]{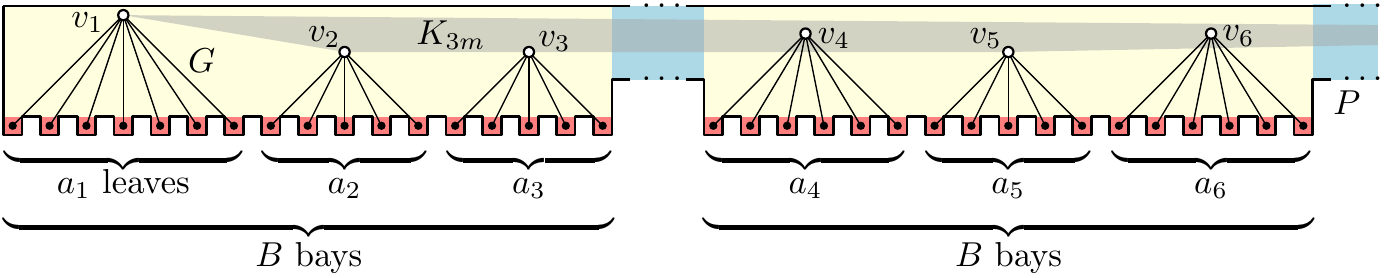}
      \caption{Idea of the reduction from \textsc{3-Partition}}
      \label{fig:hardness}
  \end{figure}
  \begin{figure}
    \centering
    \includegraphics[scale=0.99]{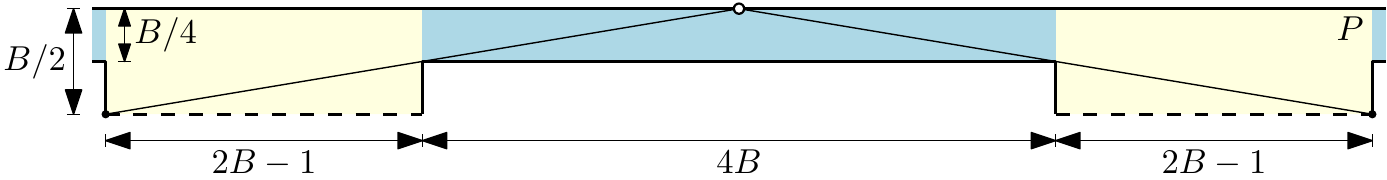}
    \caption{Size of the polygon~$P$ used in our reduction; drawing
      not (quite) to scale}
    \label{fig:polygon-size}
  \end{figure}

  If $S$ is a yes-instance, then we can sort $S$ such that 
  $\sum_{i=1}^3 a_{i}=\dots=\sum_{i=3m-2}^{3m} a_{i}=B$.
  We place the leaves adjacent to $v_1,\dots,v_{3m}$ into 
  the $Bm$ bays of~$P$, going through $P$ from left to right.  
  We place each leaf at the center of its bay.
  Then, for $i \in [3m]$, we place vertex~$v_i$ on the perpendicular 
  bisector of its leftmost leaf~$\ell_{i,1}$ and its rightmost
  leaf~$\ell_{i,a_i}$ so that $v_i$ sees exactly the leaves 
  $\ell_{i,1},\dots,\ell_{i,a_i}$ adjacent to~$v_i$.
  This is achieved by setting the height of~$v_i$ above its leaves
  to~$a_i$ if $a_i$ is even and to $a_i-1$ otherwise.
  Recall that $B/4+1 \le a_i \le B/2-1$.
  To make sure that even the lowest clique vertices can see each other
  through the corridors, we lift all clique vertices by half a unit.
  This does not change the set of bay centers that a clique vertex can
  see.  (Note that Figure~\ref{fig:hardness} does not reflect this
  slight modification.)
  The sizes of the corridors between consecutive groups of~$B$ bays
  ensure that each clique vertex sees every other
  clique vertex but no clique vertex can see into bays of different
  groups; see Figure~\ref{fig:polygon-size} (where the dashed line
  segments mark the top edges of the bays).
  
  Given a yes-instance of the obstacle problem, that is, a placement
  of the vertices of~$G$ in~$P$ such that their visibility graph with
  respect to~$P$ is isomorphic to~$G$, we show how~$S$ can be
  3-partitioned.  First, observe that no two leaves are adjacent
  in~$G$; hence any convex region in~$P$ contains at most one leaf.
  In particular, this holds for the yellow \emph{group} regions in
  Figures~\ref{fig:hardness} and~\ref{fig:polygon-size}, for each of
  the red bays, and for the bounding box of all blue corridors.
  Therefore, all but $m+1$ leaves lie in bays.  By scaling the
  partition instance by a factor of $m+2$, we can assume that each
  leaf lies in a bay.  Due to our above observation regarding the
  visibility region of clique vertices, we can map each clique vertex
  to a group of bays.  Since we assumed that, for each $i \in [1,3m]$,
  $v_i \in (B/4,B/4)$, exactly three clique vertices are mapped to
  each group of bays.  Since we can assume that every leaf lies in its
  own bay, the total number of leaves per group cannot exceed~$B$.  On
  the other hand, we must distribute a total of $mB$ leaves to $m$
  groups, so each group must get exactly $B$ leaves.  Hence the
  numbers of leaves of the three clique vertices in each group sums up
  to~$B$, which 3-partitions~$S$ as desired.
\end{proof}

\bibliographystyle{plainurl}
\bibliography{abbrv,obstacles}

\end{document}